\documentclass[lettersize,journal]{IEEEtran}

\usepackage{amsmath,amssymb,dsfont,stfloats,color,url,bbm}
\usepackage[pdftex]{graphicx}
\usepackage{subcaption}
\usepackage{cite}
\usepackage{mathabx}
\usepackage{multirow}
\usepackage{tabu}
\usepackage{amsthm}

\usepackage{siunitx}
\usepackage{tikz} 
\usepackage{pgfplots}
\usepackage{enumitem}

\usepackage[font=footnotesize]{caption}

\pgfplotsset{compat=newest} 
\usepgfplotslibrary{units} 
\DeclareGraphicsExtensions{.eps,.pdf,.png,.jpg,.gif,.jpeg,.pstex}

\newtheorem{theorem}{Theorem}

\newtheorem{fact}{Fact}
\newtheorem{lemma}{Lemma}

\newtheorem{definition}{Definition}

\newtheorem{remark}{Remark}

\def\QED{\mbox{\rule[0pt]{1.3ex}{1.3ex}}}

\renewenvironment{proof}
        {\begin{trivlist} \item[] {\itshape Proof:\ } }
        {\hfill \QED \end{trivlist}}

\def\argmax{\mathop{\rm argmax}}

\usepackage{url}

\newfont{\bbb}{msbm10 scaled 700}

\newfont{\bb}{msbm10 scaled 1100}
\newcommand{\CC}{\mbox{\bb C}}
\newcommand{\PP}{\mbox{\bb P}}
\newcommand{\RR}{\mbox{\bb R}}

\newcommand{\EE}{\mbox{\bb E}}

\newcommand{\HH}{\mbox{\bb H}}

\newcommand{\VV}{\mbox{\bb V}}

\newcommand{\yy}{\mathbbm{y}}
\newcommand{\xx}{\mathbbm{x}}
\newcommand{\zz}{\mathbbm{z}}
\newcommand{\sss}{\mathbbm{s}}

\newcommand{\hh}{\mathbbm{h}}

\newcommand{\vvv}{\mathbbm{v}}

\usepackage[mathscr]{euscript}
\newcommand{\Rs}{\mathscr{R}}
\newcommand{\Cs}{\mathscr{C}}

\newcommand{\av}{{\bf a}}

\newcommand{\hv}{{\bf h}}

\newcommand{\rv}{{\bf r}}

\newcommand{\vv}{{\bf v}}
\newcommand{\xv}{{\bf x}}

\newcommand{\zv}{{\bf z}}
\newcommand{\zerov}{{\bf 0}}

\newcommand{\Am}{{\bf A}}

\newcommand{\Fm}{{\bf F}}

\newcommand{\Qm}{{\bf Q}}

\newcommand{\Ym}{{\bf Y}}
\newcommand{\Zm}{{\bf Z}}

\newcommand{\Ac}{{\cal A}}

\newcommand{\Cc}{{\cal C}}

\newcommand{\Ec}{{\cal E}}

\newcommand{\Gc}{{\cal G}}

\newcommand{\Ic}{{\cal I}}

\newcommand{\Lc}{{\cal L}}

\newcommand{\Nc}{{\cal N}}

\newcommand{\Rc}{{\cal R}}
\newcommand{\Sc}{{\cal S}}

\newcommand{\Uc}{{\cal U}}

\newcommand{\lambdav}{\hbox{\boldmath$\lambda$}}
\newcommand{\epsilonv}{\hbox{\boldmath$\epsilon$}}
\newcommand{\nuv}{\hbox{\boldmath$\nu$}}
\newcommand{\muv}{\hbox{\boldmath$\mu$}}

\newcommand{\phiv}{\hbox{\boldmath$\phi$}}

\newcommand{\Sigmam}{\hbox{\boldmath$\Sigma$}}

\newcommand{\Omegam}{\hbox{\boldmath$\Omega$}}

\newcommand{\trace}{{\hbox{tr}}}

\renewcommand{\arg}{{\hbox{arg}}}

\newcommand{\eqdef}{\stackrel{\Delta}{=}}
\newcommand{\defines}{{\,\,\stackrel{\scriptscriptstyle \bigtriangleup}{=}\,\,}}

\newcommand{\herm}{{\sf H}}

\newcommand{\transp}{{\sf T}}

\newcommand{\SINR}{{\sf SINR}}
\newcommand{\SNR}{{\sf SNR}}

\newcommand{\Ktot}{K} 
\newcommand{\Kact}{K_{\rm act}}

\newcommand{\gammadyn}{\widehat{\gamma}}
\usepackage{stmaryrd} 

\usepackage{hyperref}
\hypersetup{
	bookmarks=true,         
	unicode=false,          
	pdftoolbar=true,        
	pdfmenubar=true,        
	pdffitwindow=false,     
	pdfstartview={FitH},    
	pdfnewwindow=true,      
	colorlinks=true,       
	linkcolor=red,          
	citecolor=green,        
	filecolor=blue,      
	urlcolor=blue           
}

\usepackage{siunitx}

\allowdisplaybreaks

\renewcommand{\arg}{{\rm arg}}

\pgfplotsset{
	kurze Legende/.style={%
		legend image code/.code={
			\draw[##1,line width=0.6pt]
			plot coordinates {
				(0cm,0cm)
				(0.25cm,0cm)
			};%
		}
	}
}

\pgfplotsset{
	normale Legende/.style={%
		legend image code/.code={
			\draw[##1,line width=1.2pt]
			plot coordinates {
				(0cm,0cm)
				(0.4cm,0cm)
			};%
		}
	}
}

\pgfplotsset{
	sehr kurze Legende/.style={%
		legend image code/.code={
			\draw[##1,line width=0.7pt]
			plot coordinates {
				(0cm,0cm)
				(0.17cm,0cm)
			};%
		}
	}
}

\begin{document}

\setlength{\abovedisplayskip}{1pt}
\setlength{\belowdisplayskip}{1pt}
\setlength{\abovedisplayshortskip}{1pt}
\setlength{\belowdisplayshortskip}{1pt}

\title{Fairness Scheduling in User-Centric Cell-Free Massive MIMO Wireless Networks}

\author{\IEEEauthorblockN{Fabian G\"ottsch\IEEEauthorrefmark{1},
		Noboru Osawa\IEEEauthorrefmark{2}, Issei Kanno\IEEEauthorrefmark{2}, Takeo Ohseki\IEEEauthorrefmark{2}, Giuseppe Caire\IEEEauthorrefmark{1}}
	\thanks{\IEEEauthorrefmark{1}Technical University of Berlin, Germany. Emails: \{fabian.goettsch, caire\}@tu-berlin.de}
\thanks{\IEEEauthorrefmark{2}KDDI Research Inc., Japan. Emails: \{nb-oosawa, is-kanno, ohseki\}@kddi.com}}

\maketitle

\vspace{-1cm}

\begin{abstract}
We consider a user-centric cell-free massive MIMO wireless 
network with $L$ remote radio units, each with $M$ antennas, serving $\Ktot$ single-antenna user devices (UEs). 
Most of the current literature considers the regime $LM \gg \Ktot$, where the $K$ UEs are active on each time-frequency slot, and evaluates the system performance in terms of {\em ergodic rates}. 
In this paper, we take a quite different viewpoint. We observe that the regime of $LM \gg \Ktot$ 
corresponds to a lightly loaded system with low sum spectral efficiency (SE). In contrast, in most relevant scenarios, the number of 
UEs is much larger than the total number of antennas (think of a sport event with $K \sim 10,000$ users and $ML \sim 200$ antennas).
To achieve high sum SE and handle $\Ktot \gg ML$, users must be scheduled over the time-frequency resource. 
The number of active users $\Kact \leq \Ktot$ must be carefully chosen such that: 1) the network operates close to its maximum SE; 
2) the active user set must be chosen dynamically over time in order to enforce fairness in terms of per-user time-averaged 
{\em throughput rates}. The fairness scheduling problem is canonically formulated
as the maximization of a suitable concave componentwise non-decreasing {\em network utility function} of the per-user rates. 
The intermitted user transmission due to scheduling imposes slot-by-slot coding/decoding, which in turn
prevents the achievability of ergodic rates. 
Hence, we model the per-slot service rates using information outage probability.
In order to obtain a tractable problem, we make a ``decoupling''  assumption on the CDF of the instantaneous mutual information seen at each UE $k$ receiver. We approximately enforce this condition by introducing a conflict graph that prevents the simultaneous 
scheduling of users with large pilot contamination conflict and propose an adaptive scheme for instantaneous 
service rate scheduling based on locally estimating the mutual information CDF at each UE.
Overall, the proposed dynamic scheduling is the first to address such system dimensions with tens of thousand users in a scalable way, is robust to system model uncertainties, and can be easily implemented in practice.
\end{abstract}

\begin{IEEEkeywords}
	User-centric, cell-free massive MIMO, fairness, scheduling, information outage probability.
\end{IEEEkeywords}


\section{Introduction} \label{intro}

Multiuser Multiple-Input Multiple-Output (MU-MIMO) has been widely investigated from a theoretical viewpoint 
\cite{Caire-Shamai-TIT03,Viswanath-Tse-TIT03,Weingarten-Steinberg-Shamai-TIT06,Caire-Jindal-Kobayashi-Ravindran-TIT10} and has become a cornerstone technology to achieve high spectral efficiency and serve a large number of user equipments (UEs) in both cellular \cite{3gpp38211,Larsson-book,cell-free-fnt} and local area wireless networks (see  \cite{khorov2018tutorial,qu2019survey} and references therein). 
Massive MIMO is a convenient implementation of MU-MIMO where the number of base station (BS) 
antennas $M$ is much larger than the number of simultaneously served UEs 
\cite{marzetta2010noncooperative,Larsson-book}. 
While massive MIMO was originally proposed for cell-based systems with per-BS processing \cite{marzetta2010noncooperative,hoydis2013massive,huh2012achieving}, more recently the concept of 
cell-free user-centric networks has been promoted in order to provide a more uniform service to very densely packed users 
\cite{ngo2017,nay2017,bjo2020,cell-free-fnt,miretti2022, goettsch2022}. 
In this paper we refer to the disaggregated reference model of 3GPP \cite{3gpp38300} (see Fig.~\ref{fig-network}) 
and consider a system formed by $L$  radio units (RUs), each equipped with $M$ antennas, 
connected to decentralized processing units (DUs) via a flexible fronthaul network. 
The RUs implement basic low-level PHY functions, such as FFTs/IFFTs for OFDM modulation,
A/D and D/A conversion, and baseband/RF modulation/demodulation. The DUs implement the high-level PHY functions such as MIMO precoding/detection and coding/decoding of the individual user data streams. 
Higher layer functions such as user-centric cluster formation, pilot and power allocation, and dynamic scheduling,
are implemented at a higher hierarchical level by one or more centralized units (CUs). 
We focus on {\em scalable} systems as defined in \cite{bjo2020,cell-free-fnt}, where each UE $k\in [\Ktot]$\footnote{We denote the set of the first positive $N$ integers by $[N] = \{ 1, \dots, N \}$.} 
is associated to a user-centric finite size set $\Cc_k$ of RUs and each RU $\ell \in [L]$ serves a finite size set $\Uc_\ell$ 
of UEs.  

\subsection{Motivation}  \label{motivation}

Early works on cell-free massive MIMO assumed $M = 1$ and $L  > \Ktot$ \cite[Ch. 2]{cell-free-fnt}.
Recognizing that the placement of ``more RUs than UEs'' is hardly justifiable from an operator deployment cost viewpoint, more recent works  have considered a more realistic RU/UE density 
regime $L < \Ktot < LM$ with $M > 1$  (e.g., see \cite{bjo2020,cell-free-fnt,miretti2022, goettsch2022,chen2022}). 
In these works, the $\Ktot$ UEs are all simultaneously active and the system performance is studied in 
terms of the per-user {\em ergodic rates} (e.g., see \cite{Larsson-book,bjo2020,cell-free-fnt,miretti2022,goettsch2022}).  However, the achievability of ergodic rates assumes that coding can be performed over a sufficiently large sequence of independent channel fading states, implying continuous transmission over many time-frequency slots. This assumption may be incompatible with dynamic scheduling and per-slot coding/decoding, 
as well as with the ``low latency'' requirement, which is as a key feature of 5G \cite{itu2017requirements, 7529226}.

\begin{figure}[t!]
\centerline{\includegraphics[width=\linewidth]{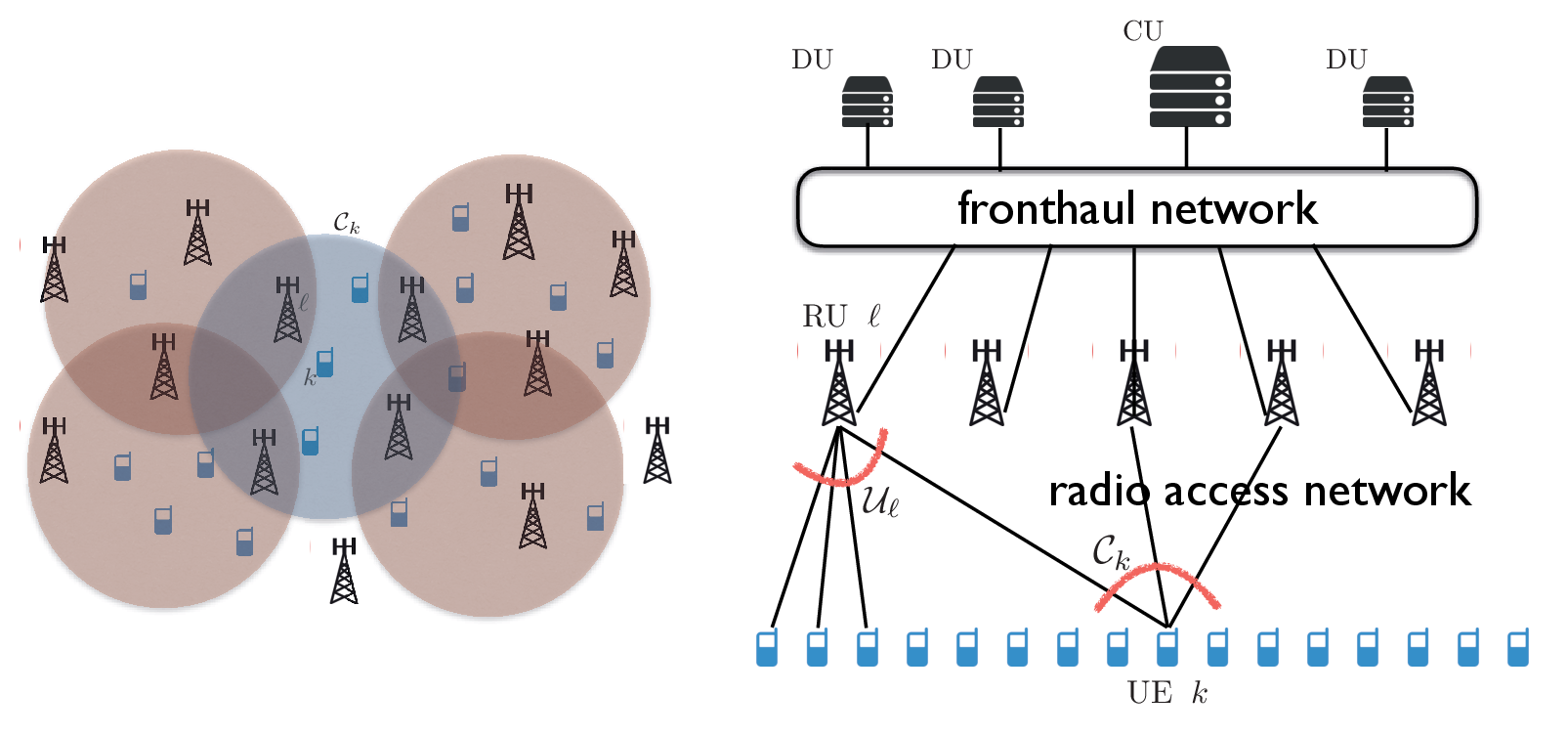}}
\vspace{-.3cm}
\caption{Left: a sketch of the distributed cell-free user-centric network. Right: a sketch of the disaggregated 
network architecture.  The cluster processors are hosted in the DUs. The CU hosts centralized (low complexity) processes, such as the user scheduler studied in this paper.}
\label{fig-network}
\vspace{-.4cm}
\end{figure}

We claim that the regime $\Ktot < LM$ with continuous transmission 
is not a practical regime of interest for such networks. 
Instead, a much more relevant regime is $\Ktot \gg LM$ where users must be scheduled over
the time-frequency resource. 
Furthermore, we also claim that the per-user ergodic rate is not a significant performance metric in this scenario. 
Instead, the per-user (long-term average) {\em throughput rate} is a much more meaningful metric. 
In this sense, we must distinguish between the {\em instantaneous rate} scheduled to the $\Kact < \Ktot$ 
active users on each given scheduling slot, and the {\em throughput rate} that each user accumulates as a time-average over a long sequence of slots. It is clear that some fairness criterion must be imposed such that all users get a chance to be 
scheduled over time and achieve a non-zero  throughput rate, even though only a subset of active users is served with a positive instantaneous rate on each scheduling slot.  
%

\begin{figure}[t!]
	\centering
	\begin{subfigure}{0.49\linewidth}
		\includegraphics[width=\linewidth]{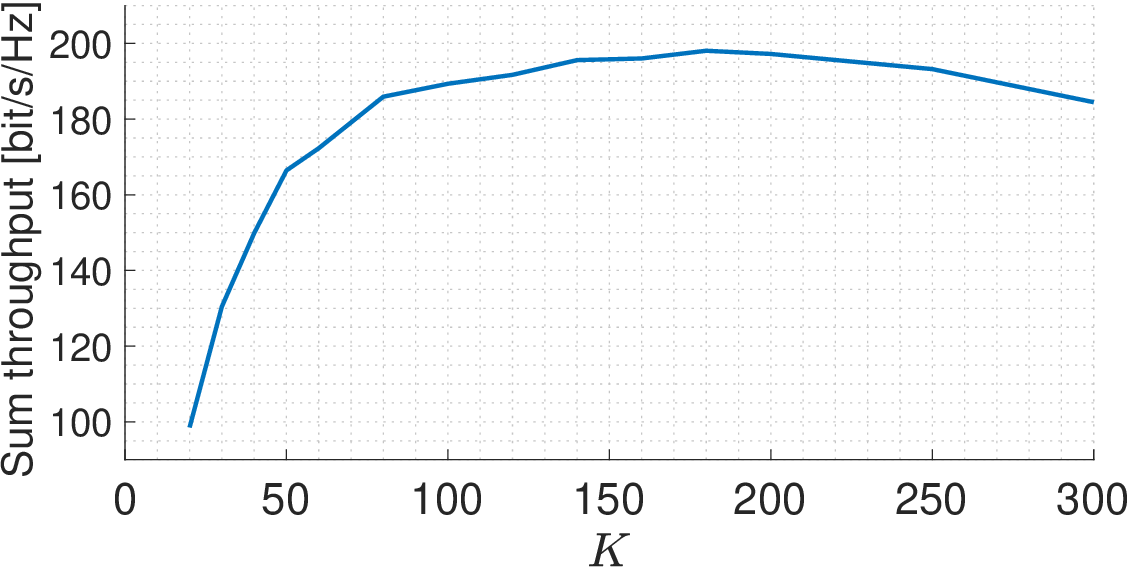}
		\caption{} 
	\end{subfigure}%
	\begin{subfigure}{0.49\linewidth}
		\includegraphics[width=\linewidth]{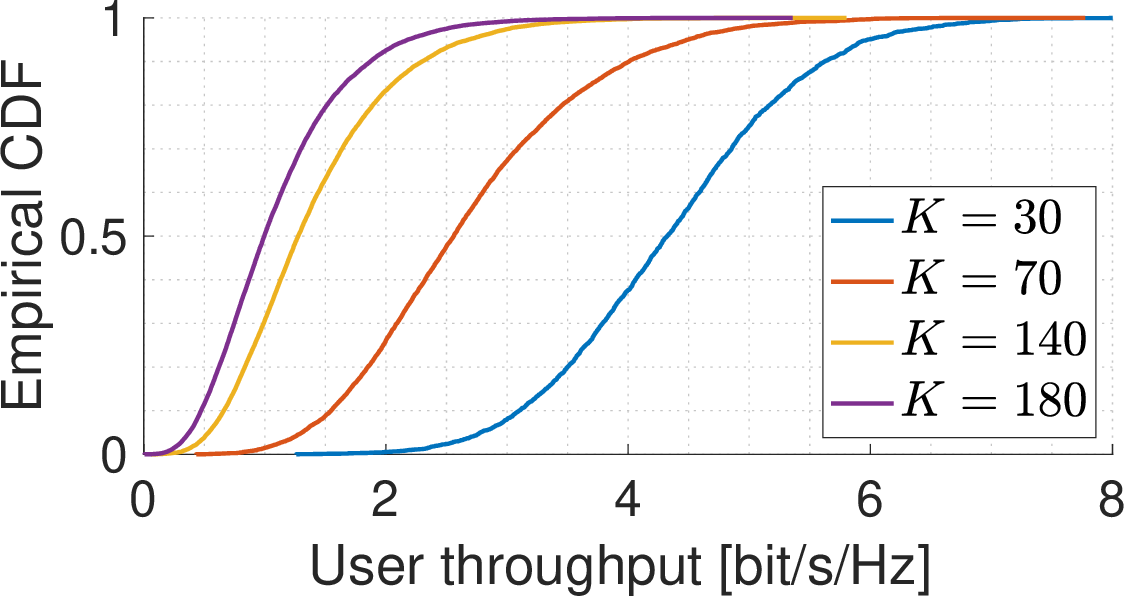}
		\caption{} 
	\end{subfigure}
	\vspace{-.2cm}
	\caption{Sum throughput rate and the empirical CDF of the user throughput rate for different $\Ktot$, where all other system parameters are fixed.} \label{fig:sumSEvsK}
	\vspace{-.4cm}
\end{figure}
To motivate the problem, consider Fig.~\ref{fig:sumSEvsK} (a),  showing the total spectral efficiency (SE), i.e., the sum throughput rate, in bits per channel use (or bit/s/Hz) for the 
reference system described in Section~\ref{sec:simulations} with $LM = 200$, where $L=20$, $M=10$, as a function of the number of users $\Ktot$, when 
all users are continuously active and encoding/decoding is done on blocks of $F = 10$ resource blocks (RBs) in frequency.
We notice that for $\Ktot$ much smaller than $LM$ the SE is small and grows linearly with $\Ktot$. 
Interestingly, most of the current literature has focused on this {\em lightly loaded} regime of low total SE and relatively 
high per-user rates  (see the corresponding cumulative distribution function (CDF) in 
Fig.~\ref{fig:sumSEvsK} (b)).  As $\Ktot$ increases, the SE ``flattens out'' and reaches its peak. 
Near the maximum SE, the per-user throughput rates collapse (see the corresponding CDF in Fig.~\ref{fig:sumSEvsK} (b)).
If $\Ktot$ keeps increasing, the SE slightly decreases, showing that in this regime the system becomes
congested.\footnote{Of course, mathematically, if one can optimize the system with full statistical information and just care about the SE, the optimized SE cannot decrease as $\Ktot$ increases because the optimal scheme would allocate power in a way that eventually some users will get almost zero power and therefore ``disappear'' effectively from the system. However, the point we want to make here is that if we insist on equal transmission power per UE in the UL or per DL stream, and keep serving more and more users, the ``naive system'' not only reaches an interference limited regime, but also decreases its overall SE performance since it enters a congested regime.}  
Hence, a good scheduler should choose the number of active users $\Kact$ (slightly) on the left of the SE peak, i.e.,  at the end of the ``linear regime'' of SE vs. the number of active users. 
Beyond this point, a marginal increase of the SE is achieved at the cost of a large fraction of users with very small rates, 
which does not make sense from a practical service viewpoint. Extensive system simulation shows (see also our short conference version \cite{gottsch2022fairness} of this work) that $\Kact \approx \frac{LM}{2}$ is usually a good choice, 
and the exact number is a design parameter that depends on the level of frequency diversity $F$, and the desired tradeoff between sum SE and per-user throughput rate.

In order to gain intuition into the problem at hand, consider for example a system where the total bandwidth is partitioned into 
100 frequency resource blocks (RBs) per time slot, serving a total number of $10,000$ users with $L = 20$ RUs with $M = 10$ antennas each (e.g., see the real-world deployment in \cite{7421132}). 
Every active user is allocated a block of $F = 10$ RBs in frequency to achieve a certain level of frequency diversity. 
Thus, the scheduler dynamically chooses on every slot a set of $\Kact \approx \frac{LM}{2} = 100$
users out of $1,000$ per RB in order to 
exploit the total system spatial degrees of freedom. 
The scheduler must also allocate an ``instantaneous'' rate to each active user since encoding/decoding is performed 
block by block, i.e., coding over a virtually infinite sequence of fading states is not possible.  
In this case, the instantaneous rate must be scheduled according to the notion of {\em information outage rate} (e.g., see \cite{biglieri1998fading}), where a non-vanishing block error probability is taken explicitly into account.


\subsection{Novelty and contributions}  \label{contributions}

In light of the above motivation, we study the fairness scheduling problem for a cell-free user-centric wireless network as defined in  \cite{bjo2020,cell-free-fnt,miretti2022,goettsch2022}. We consider a full buffer model and 
canonically formulate fairness in terms of 
the maximization of a suitably defined network utility function, i.e., a concave componentwise non-decreasing function of the user throughput rates \cite{georgiadis2006resource,neely2008fairness,shirani2010mimo}. In particular, we consider {\em proportional fairness} (PF) and {\em hard fairness} (HF), which are special cases of the family of so-called 
$\alpha$-fairness utility functions \cite{mo2000fair}.
The network utility maximization (NUM) problem is solved using the Lyapunov drift-plus-penalty (DPP) approach, which 
naturally yields a dynamic scheduling scheme \cite{georgiadis2006resource}. 
In addition, the motivation for using the proposed Lyapunov DPP framework is that, as will be shown later, with appropriately chosen parameters the framework allows to approximate the \textit{optimal} solution of the NUM as closely as desired.
		
	A seemingly analogous problem in combination with the Lyapunov DPP approach is considered in \cite{chen2019dynamic}, which is the only other work using the the Lyapunov DPP for scheduling in cell-free massive MIMO to the authors' knowledge. However, we note the following major differences. The PHY rate allocation in \cite{chen2019dynamic} is posed as a weighted sum rate maximization (WSRM) problem and in particular as a power allocation problem, where the allocated (ergodic) user rates are assumed to \textit{always} be achieved. Hence, there is a zero probability of a packet error. However, as already said, 
	ergodic rates are incompatible (and in general not achievable) under slot-by-slot decoding and the
	``low latency'' requirement.
	Moreover, applying the scheme in \cite{chen2019dynamic} (all $\Ktot$ users active and perform power control over all $\Ktot$ users) to systems of the size considered in this paper, the complexity would be enormous, making real-time scheduling computationally very challenging.
	Also, in  \cite{chen2019dynamic} the active user set is defined before solving the WSRM 
	by assuming that all users with non-empty queues transmit. This is possible for very lightly loaded systems as considered in \cite{chen2019dynamic},
	 with $100$ RU antennas and only $K = 10$ users in the network, which is definitely not the scenario considered in this work. Due to these differences, the scheme proposed in  \cite{chen2019dynamic} is not applicable to the problem considered in our work. No other work in the cell-free massive MIMO literature has provided results addressing such a large network
	(see the example above and the numbers in the simulation section) in a scalable and efficient way.

The application of the general Lyapunov DPP framework to our specific problem is non-trivial, because the instantaneous service rates, i.e., the rates scheduled to the active users, are random variables due to the finite diversity achieved in time and frequency (due to slot-by-slot coding) and in space (due to the limited number of spatial degrees of freedom
which prevents massive MIMO ``channel hardening'' \cite{8379438}). 
Scheduling with random instantaneous rates can be handled by using information outage probability as in \cite{shirani2010mimo}. This, however, requires the knowledge of the individual CDF of the mutual information at each user receiver, which in turn depends on the scheduling decision, i.e., of the selected 
set of active users. 

In order to obtain a tractable problem, we make a critical ``statistical decoupling'' assumption: 
namely,  that the CDF of the mutual information at each user $k$ receiver depends only on the local channel 
statistics of UE $k$. This assumption holds true in the massive MIMO limit of $M \rightarrow \infty$ with 
constant $\Kact/M$, under  certain conditions on the channel statistics  \cite{shirani2010mimo,Huh11,bethanabhotla2015optimal}. For systems with finite number of RU antennas 
$M$, the statistical decoupling assumption approximately holds when there is sufficient ``self-averaging'' of the 
interference term at the denominator of the {\em Signal to Interference plus Noise Ratio} (SINR) at each user decoder input. 
Intuitively, self-averaging holds when the multiuser interference is due to many small contributions. Hence, we 
enforce this condition by imposing that no two users causing strong mutual interference 
can be scheduled at the same time (this translates in a precise mathematical condition as seen later). 
To this purpose, we propose a novel {\em conflict graph} approach that prevents such ``colliding'' users to be scheduled together.  Furthermore, since the mutual information CDF is intractable in closed form, we propose to adaptively estimate it 
at each UE $k$ over a sliding window of past slots.  The locally learned empirical CDF is used to schedule the instantaneous rate of the active users.  Notice that rate adaptation based on the empirical statistics (e.g., RSSI, block error rate, etc.) 
collected over a sliding window of past slots is currently implemented in real systems. 
Thus, our approach can be seen as an information-theoretic version of such practical schemes. 

The spread of the mutual information CDF is reduced as the frequency diversity order $F$ increases.
Our analysis can capture the effect of the finite frequency diversity order $F$ on the overall system performance, 
unlike the ergodic rate analysis, where the finite diversity order effect is completely lost. 
With the knowledge of the empirical mutual information CDF for each user, 
the scheduler solves at each slot a constrained maximization of the weighted sum outage rate, where 
the weights are recursively calculated as the backlogs of {\em virtual queues} and the constraint is expressed by the conflict graph.  The problem takes the form of an integer linear program, which can be efficiently solved with standard tools 
even for fairly large systems. 

In \cite{georgiadis2006resource}, two fundamental scheduling problem formulations are provided : 1) queue stability with exogenous arrivals and transmission (not virtual) queues corresponding to a data buffer, 2) throughput fairness with ``infinite buffer'' and virtual queues. In this work we consider the second framework, where the ``infinite buffer'' assumption implies that each scheduled
user can always transmit as much data as the PHY allows. The virtual queues do not correspond to actual queued data, 
but are iteratively updated ``weights'' of the scheduling algorithm. We would like to remark that, in reality, 
no user has a truly infinite buffer, while for typical Internet applications the purely random arrival traffic is also unrealistic. 
Typically, data are transmitted/requested in bulk, e.g., for uploading a picture, or streaming a video. Hence, 
on the time scale of a few tens of seconds, the infinite buffer model for such applications is relevant.
For example, the famous proportional fairness scheduling in 1xEV-DO systems \cite{qualcomm2001pfs} was motivated by such argument. In this context, the infinite buffer is a mathematical abstraction to make the statement mathematically precise.

We also notice that well-known problems such as pilot allocation and user-centric cluster formation have been treated 
in the current literature in the assumption that all users are active all the time
\cite{cell-free-fnt,miretti2022,goettsch2022,chen2022}. In contrast, in the presence of dynamic scheduling, an important question is whether pilot assignment and cluster formation should be performed 
{\em at each scheduling decision} (dynamic reassignment) or once for all to all users independently of the scheduling decisions
(fixed assignment).  
Notice that the dynamic pilot assignment is supported by the current 
5GNR standard, where the demodulation reference signal (DMRS) for channel estimation is associated with the physical uplink shared channel \cite{3gpp38211}. A fixed DMRS pilot assignment to users is not considered in the current 5GNR standard, but could be implemented in order to decrease the control signaling overhead. 
In this work, we compare the {\em reassignment} and the {\em fixed assignment} options in terms of 
their SE performance and show that a well-designed system with fixed assignment does not suffer large degradation with respect to the more complex and overhead-demanding reassignment. 

\subsection{Outline}
In the next Section, we will describe the physical layer system model. Section \ref{sec:scheduling} introduces the information outage rate and the fairness scheduling framework. The algorithmic solutions to approach the scheduling problem and practical aspects for implementation are presented in Section \ref{sec:algorithmic_practical_solutions}. Numerical results are shown in Section \ref{sec:simulations}.

\section{Physical layer system model} \label{sec:system}

In order to present the fairness scheduling problem and the proposed solution, we first need to review a 
rather standard cell-free user-centric system as described, with minor variations, in most literature (e.g., see
\cite{ngo2017,nay2017,bjo2020,cell-free-fnt,miretti2022, goettsch2022}). 
In particular, here we follow the notation of our previous work \cite{goettsch2022}. 

The system operates in TDD mode with $L$ RUs, each equipped with $M$ antennas, 
and $\Ktot$ single-antenna UEs. Both RUs and UEs are distributed on a squared region on the 2-dimensional plane. 
Without loss of generality, we focus on a system subband\footnote{As in the example of Section~\ref{motivation}, 
$\Ktot$ can be thought as the total number of users per subband. In general, the system bandwidth may be 
an integer multiple of such subband. Hence, the total number of users in the system is a corresponding integer multiple of $\Ktot$.}
formed by $F$ frequency-domain RBs and assume 
the standard block fading model \cite{marzetta2010noncooperative}, for which  each RB is formed by $T$ 
time-frequency symbols over which the channel small-scale fading coefficients are constant and mutually independent 
for different RBs (in time and frequency) and different users.   
We let $\HH (t, f) \in \CC^{LM \times \Ktot}$ denote the overall channel matrix between all the $\Ktot$ UE antennas and all the $LM$ 
RU antennas on a given RB $f \in [F]$ in time slot $t$. The channel matrix is an $L \times \Ktot$ block matrix with 
$M \times 1$ blocks $\hv_{\ell,k}(t, f) \sim \Cc \Nc \left( \zerov, \Sigmam_{\ell,k} \right)$, each representing the 
channel vector between the $M$ antennas of RU $\ell$ and UE $k$.\footnote{Here, $\zerov$ indicates an all-zero
vector of appropriate dimension and $\Sigmam_{\ell,k} = \EE[ \hv_{\ell,k}(t, f) \hv^\herm_{\ell,k}(t, f)]$ is the 
$M \times M$ covariance matrix. Notice that the statistics of the channel vectors are 
independent of time and frequency by the well-known widely adopted 
wide-sense stationary assumption \cite{bjo2020,cell-free-fnt,miretti2022, goettsch2022}.}  
The random Gaussian vectors $\hv_{\ell,k} (t, f)$ are i.i.d. over different $t$ and $f$ and mutually independent (but not identically distributed) for different $\ell$ and $k$. 
For later use, we let $\beta_{\ell,k} = \frac{1}{M} \trace \left( \Sigmam_{\ell, k} \right)$ denote the large-scale fading coefficient (LSFC), and $\Fm_{\ell,k}$ denote  the tall unitary matrix spanning the channel dominant subspace, i.e., its columns are given by the (unit-norm) eigenvectors of $\Sigmam_{\ell, k}$ corresponding to the ``largest'' eigenvalues, as defined in \cite{adhikary2013joint} and made precise in Section~\ref{sec:simulations}, where we consider a particular antenna correlation model. 

In time slot $t$ and for all RBs $f \in [F]$, each UE $k$ is associated to its user-centric  
cluster $\Cc_k(t) \subseteq [L]$ of RUs. Consequently,  each 
RU $\ell$ is associated to a set of UEs $\Uc_\ell(t) \subseteq [\Ktot]$. The UE-RU association is described by 
a bipartite graph $\Gc(t)$ (e.g., see Fig.~\ref{fig-network}), 
which may evolve in time depending on the association scheme. 
The graph has two classes of nodes (UEs and RUs) such that the neighborhood of UE-node $k$ 
is $\Cc_k(t)$  and the neighborhood of RU-node $\ell$ is $\Uc_\ell(t)$. 
The set of edges of $\Gc(t)$ is denoted by $\Ec(t)$, i.e., $\Gc(t) = \Gc([L], [\Ktot], \Ec(t))$.

\subsection{Channel State Information} \label{sec:ch-est}

For some given scheduling policy (to be specified later), we let  $\Ac(t)\subseteq [\Ktot]$ denote the set of active users scheduled in slot $t$. 
At each time slot $t$, each RU $\ell$ obtains estimates $\widehat{\hv}_{\ell,k} (t, f)$ for all $k \in \Uc_\ell(t) \cap \Ac(t)$ 
and $f \in [F]$ from pilot sequences sent by the UEs in the UL.  
A codebook of $\tau_p$ {\em orthogonal} pilot sequences $\{\phiv_j : j \in [\tau_p]\}$ is used for channel estimation. 
This requires that $\tau_p$ signal dimensions per  block of $T$ symbols are used for UL pilots, yielding  
a SE penalty factor  $(1-\frac{\tau_p}{T})$. Pilot sequences are normalized such that $\|\phiv_j \|^2 = \tau_p \SNR$ 
for all $j \in [\tau_p]$, where the parameter $\SNR$ can be understood as the signal-to-noise ratio (SNR) at the transmitter, denoting the average transmit energy per time-frequency symbol normalized to the thermal noise power spectral density $N_0$. Each active user
$k \in \Ac(t)$ is given a pilot index $p_k(t)$ and transmits sequence $\phiv_{p_k(t)}$ over all $F$ RBs in frequency.  
The UL pilot field received at RU $\ell$ on RB $f$ in slot $t$ is given by the $M \times \tau_p$ matrix 
$\Ym_\ell^{\rm pilot} (t, f) = \sum_{i \in \Ac(t)} \hv_{\ell,i} (t, f) \phiv_{p_i(t)}^\herm + \Zm_\ell^{\rm pilot} (t, f) \label{Y_pilot}$,
where $\Zm_\ell^{\rm pilot} (t, f)$ is additive white Gaussian noise (AWGN) with elements i.i.d. $\sim \Cc\Nc(0, 1)$.

For each UE $k \in \Uc_\ell(t) \cap \Ac(t)$, RU $\ell$ employs {\em pilot matching} (by right-multiplication of the pilot field by 
$\phiv_{p_k(t)}$) and subspace projection on $\Fm_{\ell,k}$ to obtain the channel estimate \cite{goettsch2022}
\begin{align}
	\widehat{\hv}_{\ell,k} (t, f) &= \Fm_{\ell,k}\Fm_{\ell,k}^\herm \left( \frac{1}{\tau_p \SNR} \Ym^{\rm pilot}_\ell (t, f) \phiv_{p_k(t)} \right) \nonumber \\
	& =  \hv_{\ell,k} (t, f)  + \Fm_{\ell,k}\Fm_{\ell,k}^\herm \sum_{ i \in \Ac_{p_k}(t) } \hv_{\ell,i} (t, f)  \nonumber \\
	& \hspace{0.5cm} + \Fm_{\ell,k}\Fm_{\ell,k}^\herm  \widetilde{\zv}_{p_k,\ell} (t, f) ,  \label{chest1}
\end{align}
where $\Ac_{p_k}(t) = \{ i: i \in \Ac(t) \setminus k, i: p_i(t) = p_k(t) \}$, $\widetilde{\zv}_{p_k,\ell} (t, f)$ is $M \times 1$ Gaussian i.i.d. with components $\Cc\Nc(0, \frac{1}{\tau_p\SNR})$ and where the second term of the sum in  \eqref{chest1} is due to pilot contamination, i.e., it is the contribution of the active UEs transmitting the same pilot as UE $k$. The covariance matrix of the pilot contamination term is given by 
$
	\Sigmam_{\ell,k}^{\rm co}  = \sum_{i \in \Ac(t)\setminus\{k\} : p_i(t) = p_k(t)} \Fm_{\ell,k} \Fm_{\ell,k}^\herm \Sigmam_{\ell,i} \Fm_{\ell,k} \Fm^\herm_{\ell,k}$.
In particular,  when $\Fm_{\ell,k}$ and $\Fm_{\ell,i}$ are nearly mutually orthogonal, i.e. 
$\Fm_{\ell,k}^\herm \Fm_{\ell,i} \approx \zerov$, the subspace projection is able to significantly reduce the pilot contamination effect \cite{goettsch2022}. 
In most of the concurrent literature \cite{bjo2020,cell-free-fnt,miretti2022}, the channel statistics (in particular, the covariance matrices $\Sigmam_{\ell,k}$) are assumed to be known.
Schemes for channel subspace and covariance matrix estimation, respectively, in cell-free massive MIMO are presented in \cite{goettsch2022, 9715152}. In particular, the scheme in 
\cite{goettsch2022} is shown to achieve essentially the performance of ideal channel subspace knowledge.
Hence, for simplicity, in this work we assume that the subspace information $\Fm_{\ell,k}$ for all $\ell \in [L]$ and  
$k \in \Uc_\ell(t)$ is perfectly known, as justified by the results of \cite{goettsch2022}.

For convenience, we denote by $\widehat{\HH}(t, f) \in \CC^{LM \times \Ktot}$ the overall channel matrix estimated by the ensemble of the RUs. Notice that, beyond the estimation noise and pilot contamination, $\widehat{\HH}(t, f)$ differs from 
$\HH(t, f)$ by the fact that it has an $M \times 1$ all-zero block for all positions $(\ell,k)$ such that 
$k \notin \Ac(t)$ or $k \notin \Uc_\ell(t)$. This captures the fact that RU $\ell$ can obtain 
a channel estimate $\widehat{\hv}_{\ell,k} (t, f)$ only if UE $k$ is active (it is scheduled for transmission), and 
it is associated to RU $\ell$. Hence, even in the absence of estimation noise and pilot contamination, 
the RUs have a partial view of the overall channel state \cite{goettsch2022}. 

\subsection{Uplink Combining and Downlink Precoding}  \label{sec:precoding}

In this paper we consider the local linear MMSE detection and cluster-level combining scheme of  \cite{9593169,goettsch2022}.
RU $\ell$ locally computes a UL receiver combining vector $\vv_{\ell,k} (t, f)$ 
for each associated active UE $k \in \Uc_\ell(t) \cap \Ac(t)$ (see details in \cite{9593169,goettsch2022})
based on the available channel estimates, i.e., the $\ell$-th block row of $\widehat{\HH}(t, f)$.  
Using these local detection vectors, RU $\ell$ produces soft-output estimates of the time-frequency data symbols for each
user $k \in \Uc_\ell(t) \cap \Ac(t)$. The estimated symbols of user $k$ are sent via the fronthaul to the DU hosting the 
corresponding cluster processor. Such processor combines the signals from all RUs $\ell \in \Cc_k(t)$ to form the 
final received symbols for channel decoding. The cluster level combining coefficients, computed according to 
\cite{9593169,goettsch2022}, are denoted by $w_{\ell,k} (t, f)$. For convenience of notation, we define
the $LM \times 1$ dimensional unit-norm overall combining vector for user $k$ as
\begin{equation}
	\vvv_k (t, f)  = [ w_{1,k} (t, f) \vv^\transp_{1,k} (t, f) \ldots 
	w_{L,k} (t, f) \vv^\transp_{L,k} (t, f) ]^\transp , \nonumber
\end{equation}
where it is understood that $\vv_{\ell,k} (t, f) = \zerov$ for all $\ell \notin \Cc_k(t)$.

In the DL, we use the same vectors $\vvv_k (t, f)$ as downlink precoders with equal power allocation for all active user 
data streams.  The use of the UL combining vectors and DL precoding vectors is based on 
UL-DL duality, which holds approximately \cite{goettsch2022} or exactly \cite{bjo2020}, depending on which definition of achievable rate is used and the availability of channel and statistical information. 
In \cite{kddi_uldl_precoding}, we showed that using the UL combining vectors as DL precoding vectors 
with {\em uniform} power allocation over the DL data streams yields similar UL and DL rates. 
Hence, for the sake of simplicity, we use this method for the DL. We hasten to say that the scheduling approach 
developed in this paper can be applied to virtually any PHY and power allocation method, and the specific choice made here is
for convenience of exposition. 

\subsection{Data Transmission and Instantaneous Mutual Information} \label{sec:data-trans}

In the UL,  all active UEs transmit with the same average energy per symbol $P^{\rm ue} =  N_0 \SNR$.
The received $LM \times 1$ symbol vector at the $LM$ RU antennas for a single channel use on RB $f$ in slot $t$ of the UL is given by
\begin{equation} 
	\yy (t, f) = \sqrt{\SNR} \; \HH (t, f) \sss (t, f)   + \zz (t, f), \label{ULchannel}
\end{equation}
where $\sss (t, f) \in \CC^{K \times 1}$ is the vector
of information symbols transmitted by the UEs on RB $f$ in slot $t$ (zero-mean unit variance and mutually independent random variables) and  $\zz (t, f)$ is an i.i.d. noise vector with components $\sim \Cc\Nc(0,1)$.  
The cluster processor of user $k$ computes the estimate
$\hat{s}_k (t, f)  = \vvv_k (t, f)^\herm \yy (t, f)$ of the time-frequency symbol $s_k (t, f)$ of user $k$. 
Letting $\HH(t) \defines \{ \HH(t, f) : f \in [F] \}$ and $\vvv_k(t) \defines \{ \vvv_k(t, f) : f \in [F] \}$, 
the {\em instantaneous} mutual information between the transmitted symbol sequence
$\{ s_k(t,f) : f\in [F]\}$ and the detector soft-output sequence $\{ \hat{s}_k(t,f) : f\in [F]\}$ in slot $t$ (expressed in bits per time-frequency channel use, or bit/s/Hz) 
is a function of $\{ \vvv_k (t), \HH (t)\}$ and given by\footnote{This expression holds under the assumption that the user symbols are i.i.d. Gaussian and that the effective channel coefficients, i.e., the coefficients appearing in the numerator and denominator of the SINR expression in (\ref{UL-SINR-unitnorm}) are known at the receiver. Since these coefficients are constant over blocks of $T$ symbols, for simplicity we make such assumption here. A more detailed analysis would consider 
the relation between block error probability and rate of random codes in the non-coherent block-fading channel with input $s_k(t,f)$ and output $\hat{s}_k(t,f)$, where coefficients are constant over $F$ blocks of $T$ symbol each. While this might be possible using the techniques in  \cite{polyanskiy2010channel,yang2014quasi}, such information theoretic investigation goes well beyond the scope of this paper.}
\begin{gather}
	\Ic_k\left( \vvv_k (t), \HH (t) \right) \defines \frac{1}{F} \sum_{f = 1}^{F} \log\left( 1 + \SINR_k (t, f) \right), \label{eq:ul_mutual_inf}
\end{gather}
where
\begin{equation}
	\SINR_k (t, f) = \frac{  |\vvv_k (t, f)^\herm \hh_k (t, f) | ^2 }{ \SNR^{-1}  + \sum_{j \in \Ac(t): j \neq k} |\vvv_k (t, f)^\herm \hh_j (t, f) |^2 } \label{UL-SINR-unitnorm}
\end{equation}
and $\hh_k(t,f)$ is the $k$-th column of $\HH(t,f)$.
In the DL, with suitable normalization, an active UE $k$ receives 
\begin{equation} 
	y_k^{\rm dl} (t, f) = \hh_k (t, f)^\herm \xx (t, f) + z_k^{\rm dl} (t, f),  \label{DLchannel}
\end{equation}
where $z_k^{\rm dl} (t, f) \sim \Cc\Nc(0, \SNR^{-1})$ and where $\xx (t, f)  = \sum_{k \in \Ac(t)} \vvv_k(t,f) s^{\rm dl}_k(t,f)$ is
the $LM$-dimensional vector of precoded symbols transmitted collectively by the RUs,
with $s^{\rm dl}_k(t,f)$ denoting the (unit-variance) information symbol sent to UE $k$ at time slot $t$ and RB $f$.  
The DL SINR of user $k$ receiver is given by 
\begin{equation}
	\SINR^{\rm dl}_k (t, f) = \frac{|\hh_k (t, f)^\herm \vvv_k (t, f) |^2 }{\SNR^{-1} + \sum_{j \in \Ac(t) : j\neq k}   
	|\hh_k (t, f)^\herm \vvv_j (t, f) |^2},  \label{DL-SINR} 
\end{equation}
and the corresponding {\em instantaneous} mutual information is 
\begin{gather}
	\Ic_k^{\rm dl} \left( \VV (t), \hh_k (t) \right) \defines \frac{1}{F} \sum_{f = 1}^F \log\left( 1 + \SINR^{\rm dl}_k (t, f) \right),
	\label{eq:dl_mutual_inf}
\end{gather}
where we define the $ML \times K$ precoding matrix $\VV(t, f)$ with columns $\vvv_k(t,f)$ (all-zero columns for inactive users),
$\VV(t) = \{ \VV(t, f) : f \in [F]\}$ and $\hh_k(t) = \{\hh_k(t,f) : f\in [F]\}$. 
Since the precoding vectors have unit norm,  we have 
$\trace \left ( \EE [ \xx (t, f) \xx (t, f)^\herm ] \right ) =  |\Ac(t)| = \Kact$, where $|\cdot|$ denotes the cardinality of a set, i.e.,  the total transmit  power in the UL and DL are both equal to  $\Kact P^{\rm ue}$. 

\begin{remark} \label{inst-MI}
Expression \eqref{eq:ul_mutual_inf} (resp., \eqref{eq:dl_mutual_inf}) is referred to as ``instantaneous'' 
UL mutual information (resp., DL mutual information) because this is  the mutual information between 
symbols $s_k(t,f)$ (resp, $s_k^{\rm dl}(t,f)$) and the corresponding estimated symbols $\hat{s}_k(t,f)$ (resp.,  $y_k^{\rm dl}(t,f)$) conditional on the specific realization of
the (random) variables $\vvv_k (t), \HH (t)$ (resp., $\VV (t), \hh_k (t)$). This term is standard in the information theoretic literature on fading channels (e.g., see \cite{biglieri1998fading} and references therein), and should be distinguished from the standard {\em conditional} mutual information, which for the UL case would take the form
\begin{align}  
	&I \left( \{ \hat{s}_k(t,f) : f\in [F]\} ; \{ s_k(t,f) : f\in [F]\} |  \vvv_k (t), \HH (t)  \right) \nonumber  \\
		& = \EE \left [   \frac{1}{F} \sum\limits_{f = 1}^{F} \log\left( 1 + \SINR_k (t, f) \right) \right ].  \label{cond-MI}
\end{align}
While (\ref{cond-MI}) is a {\em deterministic quantity} that depends on the joint statistics 
of the true and estimated channels $\{ \HH (t,f), \widehat{\HH}(t,f) : f \in [F]\}$,  (\ref{eq:ul_mutual_inf}) and (\ref{eq:dl_mutual_inf}) are random variables, functions of the (random) instantaneous realization of 
$\{ \HH (t,f), \widehat{\HH}(t,f) : f \in [F]\}$. 
\hfill $\lozenge$
\end{remark}

\section{Fairness Scheduling} \label{sec:scheduling}

For convenience of exposition, we shall illustrate the scheduling problem for the UL (the application to the DL follows immediately). 
At each scheduling slot $t$, a scheduling policy must: 1) select a set of active users $\Ac(t)$; 2) select the 
coding rates $\rv(t) = \{r_k(t) : k \in \Ac(t)\}$ at which these users transmit their information. The system state 
in our case is $\{\Omegam(t)\}$ that denotes the mutual information statistics of each user available in time slot $t$ and depends on the statistics of $\{\HH (t)\}$, a stationary and ergodic matrix-valued Gaussian 
process as described in Section~\ref{sec:system}.  
A stationary scheduling policy $\gamma$ is a time-invariant function $\gamma : \Omega(t) \mapsto \{\Ac(t), \rv(t)\}$ \cite{georgiadis2006resource}, used to compute $\Ac(t)$ and $\rv(t)$ from $\Omegam(t)$. By definition, a stationary scheduling policy chooses the decision variables in slot $t$ based on the current network state (here the user mutual information statistics) and independent of the queues \cite{georgiadis2006resource}. Note that scheduling based on the mutual information statistics is a new approach compared to, e.g., \cite{georgiadis2006resource, shirani2010mimo}, where the current network state is assumed to be known or estimated, respectively.

We denote by $\Gamma$ the set of all feasible stationary policies for the system at hand, i.e., compliant with the PHY layer channel estimation, receiver/precoding vector calculation and Gaussian coding described in Section~\ref{sec:system}. 
In particular, driven by the discussion in Section~\ref{motivation} and exemplified by Fig.~\ref{fig:sumSEvsK}, we focus on the case where $\Ktot$ may be very large, and we consider policies operating in the ``good'' load regime such that $|\Ac(t)| \leq \Kact$, where  the maximum number of active users $\Kact$ is chosen to strike a good tradeoff between total system SE and per-user rates. In order to proceed, we introduce the following two key assumptions:

{\bf A1:} 
The rate allocation $r_k(t)$ to active user $k$ on slot $t$ is a function of the channel statistics but not of the 
instantaneous realization of $\{\vvv_k (t), \HH (t)\}$, which is known causally. \hfill $\lozenge$

{\bf A2:} 
For any user $k \in \Ac(t)$, the complementary CDF of the instantaneous mutual information
\begin{equation} 
P_k(r) \defines \PP( \Ic_k\left( \vvv_k (t), \HH (t) \right) > r ) \label{comp-CDF}
\end{equation}
is independent of the active user set $\Ac(t)$ but only on its size $\Kact = |\Ac(t)|$.
\hfill $\lozenge$

\begin{remark} \label{rem-assumptions}
Assumption A1 reflects the common practice in rate allocation in real-world systems \cite{3gpp38214, 3gpp38133}, where users are instructed to transmit a given Modulation and Coding Scheme (MCS) 
in a family spanning a wide range of coding rates on the basis of some ``local'' statistics accumulated in a time sliding window (e.g., RSSI, block error probability, in the past few time slots). The sliding window approach is used to track statistical changes, e.g., a user moving from a position close to a RU to a position farther away. This prevents slot by slot rate adaptation
depending on the instantaneous realization of the channel small scale fading states, which would be too fast to track and
too demanding in terms of protocol overhead, to signal to the receiver the used MCS. 
In practice (see Section~\ref{sec:algorithmic_practical_solutions}), since the 
analytical characterization of $P_k(r)$ is intractable, the scheduler uses empirical statistics collected over a window of time slots.

Assumption A2 is motivated by the fact that, for a large system with many randomly distributed UEs and RUs, the cumulative interference effect of all other active users on a given active user $k$ is approximately ``self-averaging''
and is weakly dependent on which individual active users are selected. This assumption is verified exactly in certain limiting conditions and symmetric situations as for example in massive MIMO multicell-networks (e.g., see \cite{marzetta2010noncooperative,hoydis2013massive,huh2012achieving,bethanabhotla2015optimal}).
For the system at hand, A2 holds only approximately, provided that UEs with strong mutual  pilot contamination are not scheduled together \cite{gottsch2022fairness}.\footnote{Our approach here to avoid strong pilot contamination consists of a conflict condition depending on the UL pilot and the spatial channel correlation (defined later in details) of users when they have RUs in common in their user-centric clusters. Although the proposed scheme is formulated to meet the characteristics of the considered directional channel model described in Section \ref{sec:simulations}, it can be tuned and adapted to various channel assumptions and pilot assignment approaches. Because we schedule only a subset of users and have finite size RU clusters, we do not run the risk of not finding a sufficient number of users that can be active at the same time.} 
For the sake of problem tractability, we shall develop our scheduling scheme under A2, and introduce a conflict graph constraint in the active user selection problem such that A2 is effectively (approximately) satisfied. \hfill $\lozenge$
\end{remark}

\subsection{Service Rate, Throughput Region, and Network Utility Maximization} \label{sec:region}

As already mentioned, because of the discontinuous user activity due to scheduling, coding
over a long sequence of scheduling slots is
impossible or impractical.\footnote{In particular,  this is in conflict with the low latency requirements typical of 5G systems \cite{itu2017requirements, 7529226}.}
With block-by-block coding/decoding, each codeword spans a single channel state $\HH(t)$. 
In this case, the block error rate of optimal codes for the effective Gaussian channel with input 
$s_k(t,f)$ and output $\hat{s}_k(t,f)$ as defined in Section~\ref{sec:data-trans} is well approximated by the 
so-called {\em information outage probability}, i.e., the probability that the instantaneous mutual information is less than the coding rate \cite{biglieri1998fading,shirani2010mimo}. We define the instantaneous {\em service rate} $\mu_k(t)$ of user $k$ as the number of information bits per s/Hz  (i.e., normalized by the block length $T F$ in channel uses) that are effectively delivered to the receiver in slot $t$.
This is given by 
\begin{eqnarray}
	\mu_k(t)  = \begin{cases}
		(1 - \frac{\tau_p}{T}) R_k(t) ,& \text{if} \; k \in \Ac(t), \\ 
		0 ,&  \text{if } k \notin \Ac(t), \end{cases} 	 \label{eq:allocated_rate}
\end{eqnarray}
where we define the random variable $R_k(t) \defines r_k(t) \times \mathbbm{1} \left\{ \Ic_k\left( \vvv_k (t), \HH (t) \right) > r_k(t) \right\} \label{eq:def_Rk}$,
and where $\mathbbm{1} \left\{ \Sc \right\}$ is the indicator function of an event $\Sc$. 
For a given stationary policy $\gamma \in \Gamma$, 
the per-user {\em throughput rate} is the long-term time-averaged service rate, i.e., 
\begin{eqnarray}
	\bar{\mu}_k = \lim_{t \rightarrow \infty} \frac{1}{t} \sum_{\tau=0}^{t-1} \mu_k(\tau) 
	= \EE \left[ \mu_k \left( \HH, \gamma \right) \right], \label{eq:ue_throughput}
\end{eqnarray}
where $\HH$ has the same marginal statistics of $\HH(t)$ and 
$\mu_k \left( \HH, \gamma \right)$ has the same marginal statistics of 
$\mu_k(t)$ in (\ref{eq:allocated_rate}).  
The convergence of the  time average to the ensemble expectation in 
\eqref{eq:ue_throughput} is with probability 1 due to the stationarity and ergodicity of the channel state and the
stationarity of the scheduling policy \cite{georgiadis2006resource}. 

A throughput rate vector $\bar{\muv} =  \left[ \bar{\mu}_1, \dots, \bar{\mu}_{\Ktot} \right]^\transp $ is feasible 
if there exists a scheduling policy $\gamma \in \Gamma$ such that $\bar{\mu}_k \leq \EE[ \mu_k(\HH, \gamma)]$ for all $k \in [\Ktot]$. Hence, the system throughput region is \cite{georgiadis2006resource}
\begin{equation}
 	\Rs = \text{coh} \bigcup_{\gamma \in \Gamma} \left\{ \bar{\muv} \in \RR_{+}^{\Ktot} : \bar{\mu}_k \leq \EE \left[ \mu_k \left( \HH, \gamma \right) \right] , \ \forall k \right\} , \label{eq:ergodic_rate_region}
\end{equation}
where ``coh'' denotes the closure of the convex hull.  
Two important properties of $\Rs$ are \cite{georgiadis2006resource}:
\begin{enumerate}
\item Queue stability region: given the system at hand, consider stationary and ergodic {\em exogenous} 
traffic arrival processes
$\{A_k(t) : k \in [\Ktot]\}$, such that $A_k(t)$ is the number of information bits per s/Hz arriving at the transmitter of 
UE $k$ in slot $t$,  with arrival rates $\lambda_k = \EE[ A_k(t)]$. Let each UE $k$ have a transmission queue $Q_k(t)$ that evolves according to the standard dynamic equation
$Q_k(t+1) = [Q_k(t)  - \mu_k(t)]_+ + A_k(t)$,
where for any $x \in \RR$ we define $[x]_+ := \max\{x, 0\}$.
An arrival rate vector $\lambdav = (\lambda_1, \ldots, \lambda_K) \in \Rs$ can be stabilized, i.e., there exists a stationary 
scheduling policy for which all system queues are {\em strongly stable} \cite{georgiadis2006resource}, 
 if and only if there exists a strictly non-negative vector $\epsilonv$ for which $\lambdav + \epsilonv \in \Rs$.
\item Sufficiency of stationary policies: $\Rs$ cannot be enlarged by non-stationary policies.
\end{enumerate}

\begin{remark}
\label{arrivals}
In this paper, we are concerned with scheduling for fairness, rather than scheduling for queue stability as discussed above. 
With a slight abuse of notation, in the following we will use $A_k(t)$ and $Q_k(t)$ to denote {\em virtual} arrival rates and {\em virtual} queues, instead of exogenous traffic arrivals and transmission queues. As anticipated in the introduction, the virtual arrival rates are computed according to the Lyapunov DDP approach in order to drive the system to operate sufficiently close to 
the throughput rate point (i.e., the point in $\Rc$) that maximizes the network utility function. The virtual queues do not correspond to
actual bits queued for transmission. Rather, they are iteratively updated weights for the scheduler to make its instantaneous 
scheduling decision by solving at each slot $t$ the WSRM problem.  \hfill $\lozenge$
\end{remark}

\begin{definition} {\bf NUM Problem.} \label{def_NUM}
Let $g(\cdot)$ denote a concave entry-wise non-decreasing function of per-user 
throughput rates $\bar{\muv}$, whose shape (concavity) captures a desired notion of fairness. The fairness scheduling problem 
consists of finding the scheduling policy solution of the NUM:\footnote{We would like to thank an anonymous reviewer who pointed out that the NUM problem can also be seen as maximizing the utility function of {\em admitted} user data rates subject to the condition that the transmission queues are stable, when the generated data rates are outside of the stable throughput region.}
\begin{eqnarray}
	& \underset{\gamma}{\text{maximize}} & g(\bar{\muv}) , \ \ \ \ \text{subject to } \bar{\muv} \in \Rs.  \label{eq:max_utility} 
\end{eqnarray}
\hfill $\lozenge$
\end{definition}

Since $\Rs$ is convex and compact, the solution of (\ref{eq:max_utility}) always exists and it is at some point on the Pareto boundary of $\Rs$. Letting $\bar{\muv}^\star$ denote such solution, there exists a scheduling policy $\gamma^\star$ that achieves $\bar{\muv}^\star$. Finding $\gamma^\star$ by directly solving  (\ref{eq:max_utility}) is generally 
impractical. In fact, despite (\ref{eq:max_utility}) being a convex optimization problem, the constraint region 
$\Rs$ is not generally characterized by a finite number of linear inequalities (see \cite{shirani2010mimo}).  
Fortunately, the NUM problem (\ref{eq:max_utility}) can be solved  to any desired degree of accuracy in an algorithmic way by using the {\em Lyapunov DPP} framework of \cite{georgiadis2006resource}. Instead of looking for a stationary policy, the DPP approach constructs a {\em dynamic policy} that yields a long-term average throughput point arbitrarily close to the  optimum $\bar{\muv}^\star$ under mild conditions (see the assumptions in Theorem~\ref{thm:performance_guarantees}). The next section develops such an algorithmic solution for the problem at hand.

\subsection{Dynamic Scheduling Policy} \label{sec:dynamic}

We start with the following somehow obvious lemma: 

\begin{lemma} {\bf Outage Rate Allocation.}  \label{lemma-rate}
For each $k \in [\Ktot]$ define
\begin{equation} 
r_k^*  = \argmax_{r \geq 0} \; r P_k(r). \label{inst-rate-allocation}
\end{equation}
and for a given active user set $\Ac(t)$ define the vector $\rv^*(t)$ with $k$-th component equal to $r_k^*$ if $k \in\Ac(t)$ and zero otherwise.  
For any stationary policy $\gamma : \Omegam(t) \mapsto (\Ac(t), \rv(t))$ yielding the throughput rate vector 
$\bar{\muv}(\gamma)$, the stationary policy $\gamma^*$ that coincides with $\gamma$ on the active user set 
$\Ac(t)$ but uses rates $\rv^*(t)$ in (\ref{inst-rate-allocation}) yields $\bar{\muv}(\gamma^*) \geq \bar{\muv}(\gamma)$, 
where the inequality holds componentwise (not necessarily strictly).
\hfill \qed
\end{lemma}

\begin{proof}
The throughput rate of user $k$ under $\gamma$ is given by 
\begin{align}
\bar{\mu}_k(\gamma) & = \left (1 - \frac{\tau_p}{T} \right ) \times \nonumber \\
& \hspace{.45cm} \EE \left [ 
r_k(t) \; \mathbbm{1} \left\{ \Ic_k\left( \vvv_k (t), \HH (t) \right) > r_k(t) \right \} \; \mathbbm{1} \left\{ k \in \Ac(t)\right \} \right ] \nonumber \\
& =  \left (1 - \frac{\tau_p}{T} \right ) \times r_k(t) P_k(r_k(t))  \; \PP( k \in \Ac(t)) , \label{zio}
\end{align}
where (\ref{zio}) follows from the fact that the distribution of $\Ic_k\left( \vvv_k (t), \HH (t) \right)$ does not depend on
$\Ac(t)$ (Assumption A2).
Since $\gamma$ and $\gamma^*$ coincide in $\Ac(t)$ and differ only in the rate allocation, 
replacing $r_k(t)$ with $r_k^*$ in (\ref{zio}) maximizes the throughput rate. 
\end{proof}

Lemma~\ref{lemma-rate} solves the outage rate allocation problem in the sense that we can restrict to policies 
that use the coding rate $r_k^*$ in (\ref{inst-rate-allocation}) for all active users. 
Next, using the theory in  \cite{georgiadis2006resource,shirani2010mimo}, we are ready to state the dynamic scheduling policy
and prove its properties under assumptions A1 and A2. 
We associate to each UE $k \in [\Ktot]$ a virtual queue $Q_k(t)$. This queue does not represent bits that arrive at UE $k$ transmitter and wait to be delivered to the receiver, since our system has full buffer and no random arrival. 
Instead, the values $\{Q_k(t) : k \in [\Ktot]\}$ are iteratively computed and yield the weights for a WSRM problem that the scheduler solves at every slot $t$ to determine the set of active users.  The queues can be initialized as $Q_k(0) = 0$ for all $k$. 

\begin{definition} {\bf Fairness dynamic scheduling.} \label{dynamic-scheduling}
Let $V > 0$, $A_{\max} > 0$ and $\Kact \geq 1$ be the parameters of the scheduling policy. 
For each $t = 0, 1, 2, \ldots$, the policy $\gammadyn$ iterates the following steps:
\begin{enumerate}
\item Virtual arrivals: let $A_k(t) = a_k$, where $\av = \left[a_1 \dots a_K\right]^\transp$ is the solution of the convex optimization problem
\begin{eqnarray} \label{arrival-rate}
	\begin{split}
	& \underset{\av}{\text{maximize}} & & V g(\av)  - \sum_{k \in [\Ktot]} a_k Q_k(t), \\ &\text{subject to } & &\av \in [0, A_{\max}]^{\Ktot} 
	\end{split} 
\end{eqnarray}
\item User selection via WSRM: let $\Ac(t)$ be the solution of 
\begin{eqnarray}  \label{user-selection} 
	\begin{split}
		& \underset{\Ac}{\text{maximize}} & & \sum_{k \in \Ac} Q_k(t) r_k^* P_k(r_k^*), \\ &\text{subject to}  & &\Ac \subseteq [\Ktot], \;\; |\Ac| \leq \Kact 
	\end{split}
\end{eqnarray}
\item Transmission: each user $k \in \Ac(t)$ transmits with rate $r_k^*$. The cluster receivers compute the receiving vectors from the estimated channels in $\widehat{\HH}(t)$ and attempt decoding. Let 
$\mu_k(t)$ denote the resulting service rate (see (\ref{eq:allocated_rate})). 
\item Virtual queue update: for all $k \in [\Ktot]$ compute the new virtual queue state as
\begin{equation} 
Q_k(t+1) = [Q_k(t)  - \mu_k(t)]_+ + A_k(t). \label{Qk-evo-virtual}
\end{equation}
\end{enumerate}
\hfill $\lozenge$
\end{definition}
The following result establishes the performance guarantee of   scheduling policy $\gammadyn$.
\begin{theorem} \label{thm:performance_guarantees}
	Suppose that channel states $\HH(t)$ are i.i.d. over timeslots. Under assumptions A1 and A2, we consider the scheduling policy $\gammadyn$ from Definition \ref{dynamic-scheduling} and constants $V>0$, $A_{\rm max} > 0$. We further assume 
	that there exists a point $\rv \in \Rs$ with strictly positive entries such that $g(\rv/2) > - \infty$. Then:
	\begin{itemize}
		\item[a)] The utility associated with the time average transmission rates achieved by $\gammadyn$ satisfies
		\begin{gather}
			\liminf_{t\rightarrow\infty} g \left( \frac{1}{t} \sum_{\tau = 0}^{t - 1} \EE \left[ \muv(\tau) \right] \right) \geq g( \bar{\muv}^\star ( A_{\rm max} ) ) - \frac{C}{V} ,
		\end{gather}
		where $\muv(\tau) = \left[ \mu_1(\tau), \dots, \mu_{\Ktot}(\tau) \right]^\transp $ and
		\begin{align}
			& C \defines \frac{\Ktot}{2} \left( A_{\rm max}^2 \right. \nonumber \\ 
			& \; \; \; \; \; \; \left. + \EE \left[ \left( \frac{1}{F} \sum_{f=1}^{F} \log \left( 1 + \frac{ | \hh_{k^\star}(t, f) |^2 }{\SNR^{-1}} \right) \right)^2 \right]  \right)
		\end{align}
		with $k^\star = \underset{ k \in [\Ktot] }{\arg\max} \ \frac{1}{F} \sum_{f=1}^{F} \log \left( 1 + \frac{ | \hh_k(t, f) |^2 }{\SNR^{-1}} \right)$, and where $\bar{\muv}^\star ( A_{\rm max} )$ is the solution of problem (\ref{eq:max_utility}) with the additional constraint $0 \leq \bar{\mu}_k \leq A_{\rm max}$ for all $k \in [K]$.
		\item[b)] For any point $\bar{\muv} \in \Rs$ such that $0 \leq \bar{\mu}_k \leq A_{\rm max}$ for all $k$, and for any value $\beta \in \left[ 0, 1 \right]$ we have
		\begin{align}
			& \limsup_{t\rightarrow\infty} \frac{1}{t} \sum_{\tau = 0}^{t - 1} \sum_{k=1}^{\Ktot} \bar{\mu}_k \EE\left[ Q_k(\tau) \right]  \nonumber \\
			& \; \; \; \leq \frac{C + V \left[ g(\bar{\muv}^\star ( A_{\rm max} ) )  - g( \beta \bar{\muv} ) \right]}{ 1 - \beta } .  \label{queue-bound}
		\end{align}
		Hence, all the virtual queues $Q_k(t)$ are strongly stable (see definition in \cite{georgiadis2006resource}).
	\end{itemize}
\end{theorem}
{\em \ \ \ Proof:} See Appendix~\ref{app:proofs}. \hfill $\blacksquare$

Theorem~\ref{thm:performance_guarantees} implies that if $A_{\rm max}$ is sufficiently large, such that $A_{\rm max} \geq \bar{\mu}_k^\star$ for all $k$, then
$\liminf_{t\rightarrow\infty} g \left( \frac{1}{t} \sum_{\tau=0}^{t-1} \EE[ \muv(\tau) ] \right) \geq g\left( \bar{\muv}^\star \right)- \frac{C}{V}$.
Hence, by choosing a sufficiently large $V$, $\gammadyn$ can approach $g\left( \bar{\muv}^\star \right)$ as closely as desired
with convergence time that scales as $O(V)$ due to (\ref{queue-bound}). In other words, using the proposed scheduler in Definition \ref{dynamic-scheduling}, we can approximate the optimal solution of the NUM posed in \eqref{eq:max_utility} by increasing $V$ (such that the performance gap decreases as $1/V$) until the converged resulting network utility approaches asymptotically its maximum. 
As a matter of fact, when the NUM problem does not admit a direct solution, the dynamic policy itself can be used as an iterative algorithm to approximate the NUM solution numerically.

\section{Algorithmic and Practical Aspects} \label{sec:algorithmic_practical_solutions}

As in the previous section, we focus on the UL. The algorithms can easily be applied to the DL case by 
using the DL instantaneous mutual information \eqref{eq:dl_mutual_inf} and its CDF in place of its UL counterpart.
For assumption A2 to hold, we need to guarantee that the pilot contamination affecting UE $k$ is negligible. In fact, 
a strong pilot contamination severely affects the quality of the channel estimates \eqref{chest1} and this affects both the useful signal term and the interference term in the SINR expression (\ref{UL-SINR-unitnorm}). 
Severe pilot contamination occurs if users $k, k' \in \Ac(t)$ share at least one RU in their user-centric cluster, 
have the same UL pilot, and their channel subspaces with respect to the shared RUs are strongly aligned, i.e., far from 
mutually orthogonal.  

We propose to prevent strong pilot contamination between active users 
by introducing a conflict graph that constrains the scheduling decisions. 
Note that it is generally not advisable to schedule users with strong mutual interference. Hence, although this is a heuristic approach, it is expected that the conflict graph does not limit significantly the space of {\em good policies}.
Interestingly, by introducing the conflict graph as a constraint in 
the general WSRM (\ref{user-selection}), the problem becomes a linear integer program. 

In particular, we define the conflict graph $\Cs = ([\Ktot], \Ec_\Cs) $ with a vertex set corresponding to all $\Ktot$ UEs in the network and an edge set $\Ec_\Cs$ accounting for the conflicts.
A UE-pair $(k, k')$ is in conflict on slot $t$ if the following three conditions are satisfied:
\begin{enumerate}
	\item the UEs are associated to at least one common RU, i.e., $\Cc_{k, k'}(t) \defines \Cc_k(t) \cap \Cc_{k'}(t) \neq \emptyset$; 
	\item the UEs are assigned the same UL pilot, i.e., $p_k(t) = p_{k'}(t)$;
	\item the channels of the UEs with respect to at least one RU $\ell \in \Cc_{k, k'}(t)$ are strongly aligned, i.e., 
		\begin{equation}
			\lVert \Fm_{\ell,k}^\herm \Fm_{\ell,k'}  \rVert_F > \eta_\Fm ,  \;\;\; \mbox{for some} \;\; \ell \in \Cc_{k, k'}(t), 
			\label{eq:subspace_conflict}
		\end{equation}
		where $\eta_\Fm$ is a threshold for ``non-orthogonality'' and $ \lVert \cdot \rVert_F $ denotes the Frobenius norm.
\end{enumerate}
The graph $\Cs$ has an edge between the vertex $k$ and vertex $k'$ for all UE-pairs $(k, k')$ in conflict.
We consider two options for the assignment of pilots and clusters, fixed and dynamic.
With the ``fixed assignment'' scheme, pilots and clusters are assigned to each UE $k \in [\Ktot]$ and kept fixed for all time slots. 
With the dynamic ``pilot reassignment'' scheme, clusters are fixed based on the LSFCs, but pilot allocation is carried out in 
each time slot after making the active user selection.  Of course, it is expected that by reassigning pilots to the active users on each slot, 
better throughput rate performance can be achieved at the cost of a higher control signaling overhead. In contrast, 
the fixed assignment requires less control signaling overhead. While a precise quantitative analysis of the signaling overhead is
out of the scope of this paper, it is nevertheless interesting to compare the two approaches in terms of 
the achieved user throughput rate. 

\subsection{Fixed Pilots and Clusters}

Based on the conflict definition, we propose the following fixed pilot assignment and cluster formation scheme:
\begin{enumerate}
	\item When a UE $k$ joins the system, it connects to a maximum of $\Cc_{\rm max}$ RUs with the largest LSFCs, provided that $\beta_{\ell,k} \geq \frac{\eta}{M \SNR } $, forming the set $\Cc_k$.\footnote{Notice that the maximum beamforming gain of an RU array with $M$ antennas is equal to $M$. Hence, this condition imposes that the 
	SNR at UE $k$ for the signal from RU $\ell$ with best-case beamforming gain is $M \beta_{\ell,k} \SNR \geq \eta$, for some suitably chosen association threshold $\eta$.}
	\item Each RU $\ell \in \Cc_k$ considers all pilot indices $p = [\tau_p ]$. 
	If user $k$ is given pilot $p$, the set of UEs $k' \neq k$ conflicting with $k$ is given by 
	$\Cs_k (p) = \left\{ \bigcup_{\ell \in \Cc_k} \Cs_{\ell, k} (p) \right\}$, where 
	\begin{eqnarray}
			\Cs_{\ell, k} (p) = \left\{ k' \in \Uc_\ell :  p_{k'} = p,  \lVert \Fm_{\ell,k}^\herm \Fm_{\ell,k'}  \rVert_F > \eta_\Fm  \right\}.  \label{eq:conflicting_ues}
	\end{eqnarray}
In fact, these are all the users having at least one RU in common with UE $k$, and aligned channel subspaces in the sense of condition \eqref{eq:subspace_conflict}. 
	\item Then, UE $k$ is assigned the pilot $p_k =  \underset{p \in [\tau_p]}{\arg \min} \ | \Cs_k (p) | $ (if the minimizer is not unique, an arbitrary choice in the minimizing set is made). 	
\end{enumerate}
Letting for simplicity of notation $\bar{R}_k \defines r_k^* P_k(r_k^*)$, and defining a binary vector $\xv \in \{0,1\}^\Ktot$ such that the $k$-th entry of $\xv$, i.e., $x_k$, is equal to $1$ if $k \in \Ac(t)$ and $0$  if $k \notin \Ac(t)$,
the WSRM problem (\ref{user-selection}) subject to the conflict graph reduces to the linear integer program 
\begin{equation}
	\begin{array}{l l}
		\underset{\xv}{\text{maximize}} & \ \ \sum_{k \in [\Ktot]} Q_k(t) \bar{R}_k x_k  \\
		\text{subject to} & \ \ \sum_{k \in [\Ktot]} x_k \leq \Kact ,  \\
		& \ \ x_k \in \left\{ 0, 1 \right\}, \\
		& \ \ x_k + x_{k'} \leq 1, \ \forall (k, k') \in \Ec_\Cs. \label{eq:sched_conflict_graph}
	\end{array}
\end{equation}
Problem \eqref{eq:sched_conflict_graph} can be efficiently computed using standard solvers (e.g., Gurobi or Matlab), even for fairly large systems.

\subsection{Pilot Reassignment Scheme}

As for the fixed assignment scheme, each UE $k$ connects to a maximum of $\Cc_{\rm max}$
RUs with the largest LSFCs, provided that $\beta_{\ell,k} \geq \frac{\eta}{M \SNR } $, forming the set $\Cc_k$ for all time slots. 
The pilots are reassigned only to the set of active users on each time slot. 
However, in this case we have a classical ``chicken and egg problem'': on one hand, 
the set of active users must be determined in order to assign the pilots; 
on the other hand, the determination of the active user set depends on the pilot assignment.
Our pragmatic solution consists of first running an active user pre-selection by solving the unconstrained WSRM problem, 
then assigning pilots on the pre-selected users and determining the conflict graph, and finally solving the WSRM problem 
with conflict graph constraint on the restricted set of pre-selected users. 
For some $\widetilde{K} \geq \Kact$, the user pre-selection finds the set of $\widetilde{K}$ users maximizing (\ref{user-selection}) without conflict constraints, i.e., it solves 
\begin{equation}
	\begin{array}{l l}
		\underset{\xv}{\rm maximize} & \ \ \sum_{k \in [\Ktot]} Q_k(t) \bar{R}_k x_k  \\
		\text{subject to} & \ \ \sum_{k \in [\Ktot]} x_k \leq \widetilde{K},  \\
		& \ \ x_k \in \left\{ 0, 1 \right\} . \label{eq:sched_asilomar}
	\end{array}
\end{equation}
The solution of \eqref{eq:sched_asilomar} 
is immediate and consists of sorting the users in decreasing order of the product  $Q_k(t) \bar{R}_k$ and 
selecting the top $\widetilde{K}$ sorted users. We denote such set as $\widetilde{\Ac}(t)$. 
Then, following steps 2) and 3) of the fixed pilots scheme, pilots are assigned to the users 
$k \in \widetilde{\Ac}(t)$ and the corresponding conflict graph $\Cs(\widetilde{\Ac}(t), \Ec_\Cs)$ is constructed. 
Finally, the set of active users $\Ac(t)\subseteq \widetilde{\Ac}(t)$ is the solution of 
\begin{equation}
	\begin{array}{l l}
		\underset{\xv}{\text{maximize}} & \ \ \sum_{k \in \widetilde{\Ac}(t)} Q_k(t) \bar{R}_k x_k  \\
		\text{subject to} &  \ \ \sum_{k \in \widetilde{\Ac}(t)} x_k \leq \Kact,  \\
		& \ \ x_k \in \left\{ 0, 1 \right\} , \ \forall k \in \widetilde{\Ac}(t),  \\
		& \ \ x_k = 0 , \ \forall k \notin \widetilde{\Ac}(t) , \\
		& \ \ x_k + x_{k'} \leq 1 , \ \forall (k, k') \in \Ec_\Cs. \label{eq:sched_conflict_graph_dynamic}
	\end{array}
\end{equation}

\subsection{Proportional Fairness and Hard Fairness Scheduling}

As fairness criteria in this paper we consider the very well-known PF and HF, implemented using the general Lyapunov DPP framework
by selecting the appropriate network utility functions in (\ref{eq:max_utility}). 
In case of PF scheduling (PFS), the network utility function is given by 
$g(\av) = \sum_{k \in [\Ktot]} \log a_k$. The corresponding solution of (\ref{arrival-rate}) is given by \cite{shirani2010mimo}
\begin{equation}
	a_k = \min \left\{ \frac{V}{Q_k(t)}, A_{\rm max} \right\} . \label{eq:a_k_pfs}
\end{equation}
In case of HF scheduling (HFS), the network utility function is given by $g(\av) = \underset{k \in [\Ktot]}{\min} \ a_k$ 
and the corresponding solution to (\ref{arrival-rate}) is given by \cite{shirani2010mimo}
\begin{align}
	a_k = \begin{cases}  A_{\rm max} , & \text{if } V > \sum_{k \in [\Ktot]} Q_k(t), \\ 
		0 , & \text{else. }  \end{cases} \label{eq:a_k_hf}
\end{align}

\subsection{Mutual Information Statistics}

In order to compute the instantaneous rate $r_k^*$ according to \eqref{inst-rate-allocation}, the mutual information CDF 
$P_k(r)$ defined in (\ref{comp-CDF}) is needed. Unfortunately, a closed-form expression for this CDF is intractable for the system at hand, since it is very difficult to compute the instantaneous mutual information even for much simpler systems. For co-located massive MIMO, the SINR ``hardens to a deterministic limit'' when the number of antennas and the number of users grow to infinity at a certain fixed ratio, which can sometimes be computed in closed form using results from random matrix theory (see, e.g., \cite{zhang2021local, 8379438, huh2011multi}). These results however are extremely delicate with respect to the model assumptions and hard to generalize when these assumptions are not satisfied.
 
Therefore, we propose an adaptive approach where each user $k$ 
accumulates samples of the instantaneous  mutual information in a sliding window of $N$ past time slots where the user is active.\footnote{In practice, this can be done at the network infrastructure side, and the rate allocation decision can be communicated together with the scheduling decision to each user through the control information in each scheduling slot, as currently specified in the 3GPP 5GNR standard \cite{3gpp38211}.} 
Based on Assumption A2 and the strong law of large numbers, the empirical CDF of $\Ic_k\left( \vvv_k (t), \HH (t) \right)$ constructed from the $N$ samples converges to the true CDF as $N \rightarrow \infty$, and thus is a suitable approximation to compute a meaningful $r_k^*$ with (\ref{inst-rate-allocation}).\footnote{As in 
\cite{gottsch2022fairness}, the allocated rates are initialized by a ``start-up'' phase consisting of $N_{\rm init}$ time slots. In each of the $N_{\rm init}$ time slots $\Kact$ out of the $\Ktot$ UEs are randomly selected considering the conflict graph to 
be active. In practice, a user joining the system would start with a very conservative rate and progressively ``ramp up'' 
the value of $r_k$ until the maximum of the product in (\ref{inst-rate-allocation}) is achieved. 
Actual practical algorithms for rate scheduling work on averaged local statistics along these lines.} 
In stationary conditions, the instantaneous mutual information distribution is independent of the slot time $t$. In practice, with moderate user mobility, the statistics change slowly over time. 
Although not investigated in this paper, the proposed method can track non-stationary (slowly varying) statistics and in fact it is very reminiscent of practical rate allocation schemes based on some time-averaged ``channel quality indicator'' (the role of which, in our case, is represented by the instantaneous mutual information CDF). 

\begin{table}[t] 
	\begin{center}
		\begin{tabular}{ | l | p{6.6cm} |}
			\hline
			{\textbf{Parameter}} & {\textbf{Description}} \\ \hline
			$A_{\rm max}$, $V$ & Algorithm parameters \\ \hline
			$\widetilde{\Ac}(t)$, $\Ac(t)$ & Preliminary and final set of active UEs \\ \hline
			$A_k(t)$ & Virtual arrival rate \\ \hline
			$\Cs$ & Conflict graph \\ \hline
			$\Cc_k(t)$ & User-centric cluster of RUs serving UE $k$ \\ \hline
			$F$ & RBs per subchannel \\ \hline
			$\Fm_{\ell,k}$ & Subspace matrix of channel $\hv_{\ell,k}$ \\ \hline
			$g(\cdot)$ & Fairness utility function \\ \hline
			$\gamma$ & Stationary scheduling policy \\ \hline
			$\Gamma$ & Set of all feasible stationary policies \\ \hline
			$\hat{\gamma}$ & Fairness dynamic scheduling policy \\ \hline
			$\Kact$ & Number of active UEs \\ \hline
			$\widetilde{K}$ & Number of preliminary active UEs \\ \hline
			$\Ktot, L, M$ & Number of UEs, RUs, and RU antennas \\ \hline			$\mu_k(t)$ & Instantaneous service rate \\ \hline
			$\bar{\mu}_k(t)$ & Long-term throughput rate  \\ \hline
			$\widetilde{\mu}_k(t)$ & Long-term throughput rate in bit/s  \\ \hline
			$\Omegam(t)$ & Mutual information statistics of all users  \\ \hline
			$p_k(t)$ & UL pilot index of UE $k$  \\ \hline
			$P_k(r)$ & Complementary CDF of the instantaneous mutual information \\ \hline
			$Q_k(t)$ & Virtual queue state \\ \hline
			$\Rc$ & System throughput region \\ \hline
			$r_k(t)$ & Allocated rate to an active user \\ \hline
			$r_k^*$ & Allocated rate that maximizes the expected service rate \\ \hline
			$R_k(t)$ & Allocated rate to an active user under information outage \\ \hline
			$\Uc_\ell(t)$ & Set of users served by RU $\ell$ \\ \hline
		\end{tabular} 
		\label{list_simulation_params}
	\end{center}
	\caption{Parameters of the scheduling problem.}
\end{table}

\section{Numerical Results} \label{sec:simulations}

We consider a cell-free network spanning an area of $A=200 \times 200\text{m}^2$ with a torus topology to avoid boundary effects, containing $L=20$ RUs, each with $M=10$ antennas. 
A bandwidth of $W=60$ MHz and noise power spectral density of $N_0 = -174$ dBm/Hz is assumed. 
For cluster formation we have chosen SNR threshold $\eta=1$ and maximum cluster size $\Cc_{\rm max}=7$.
The LSFC statistics follow the 3GPP urban microcell street canyon pathloss model from \cite[Table 7.4.1-1]{3gpp38901}, which differentiates between UEs in line-of-sight (LOS) and non-LOS (NLOS). The probability of LOS is distance-dependent and given in \cite[Table 7.4.2-1]{3gpp38901}. 
A log-normal Gaussian random variable with different parameters for LOS and NLOS is added to the deterministic 
distance dependent term to account for shadow fading.
The UL energy per symbol is chosen such that $\bar{\beta} M \SNR = 1$ (i.e., 0 dB), when the expected 
LSFC $\bar{\beta}$ is calculated for the considered statistics at 
distance $2.5 d_L$, where $d_L = \sqrt{\frac{A}{\pi L}}$ is the radius of a disk of area equal to $A/L$. 
This leads to a certain level of overlap of the RUs' coverage areas considering the RU-UE association threshold, such that each UE is likely to be associated to several RUs. The UEs are randomly dropped in the network area, while the RUs are placed on a $4\times 5$ rectangular grid. For the rate adaptation scheme we run an initialization phase with 
$N_{\rm init} = 500$ and construct for each user the instantaneous mutual information empirical 
CDF with $N = 100$ samples.  We use RBs of dimension $T = 200$ symbols and UL pilots of dimension $\tau_p = 20$, yielding a SE penalty factor of $(1 - \frac{\tau_p}{T}) = 0.9$.

For the spatial correlation between the channel antenna coefficients, we consider a simple directional channel model defined as follows.  Consider the angular support $\Theta_{\ell,k} = [\theta_{\ell,k} - \Delta/2, \theta_{\ell,k} + \Delta/2]$ centered at angle $\theta_{\ell,k}$ of the LOS 
between RU $\ell$ and UE $k$ (with respect to the RU boresight direction), with angular spread $\Delta$. 
Let $\Fm$ denote the $M \times M$ unitary DFT matrix with $(m,n)$-elements
$\left[ \Fm \right]_{m,n} = \frac{e^{-j\frac{2\pi}{M} mn}}{\sqrt{M}}$ for  $m, n  = 0,1,\ldots, M-1$, and consider the angular support set $\Sc_{\ell,k} \subseteq \{0,\ldots, M-1\}$ 
obtained according to the single ring local scattering model \cite{adhikary2013joint}, where $\Sc_{\ell,k}$ contains the DFT quantized angles (multiples of $2\pi/M$) falling inside an interval of length $\Delta = \pi/8$ placed symmetrically around the direction joining UE $k$ and RU $\ell$. Then, the channel between RU $\ell$ and UE $k$ on RB $f$ in slot $t$ is
$\hv_{\ell,k} (t, f) = \sqrt{\frac{\beta_{\ell,k} M}{|\Sc_{\ell,k}|}}  \Fm_{\ell,k} \nuv_{\ell, k} (t, f)$,
where using a MATLAB-like notation $\Fm_{\ell,k} \eqdef \Fm(: , \Sc_{\ell,k})$ denotes the tall unitary matrix obtained by selecting the columns 
of $\Fm$ corresponding to the index set $\Sc_{\ell,k}$,\footnote{Note that for uniform linear arrays (ULAs) and uniform planar arrays (UPAs), as widely used in today's massive MIMO implementations, the channel covariance matrix is Toeplitz (for ULA) or Block-Toeplitz (for UPA), and that large Toeplitz and block-Toeplitz matrices are approximately diagonalized by DFTs on the columns and on the rows (see \cite{adhikary2013joint} for a precise statement based on Szeg\"o's theorem).} 
and $\nuv_{\ell,k} (t, f)$ is an $|\Sc_{\ell,k}| \times 1$ i.i.d. Gaussian vector with components 
$\sim \Cc\Nc(0,1)$. The corresponding covariance matrix is $\Sigma_{\ell,k} = \frac{\beta_{\ell,k}M}{\left| \Sc_{\ell,k} \right|} \Fm_{\ell,k} \Fm_{\ell,k}^\herm$. 
For the definition of the conflict graph (see (\ref{eq:subspace_conflict})) we chose the 
``non-orthogonality'' threshold $\eta_\Fm = 0$. Hence,  the ``non-orthogonality'' condition $\lVert \Fm_{\ell,k}^\herm \Fm_{\ell,k'}  \rVert_F >  0$ can be stated equivalently as $ \left| \Sc_{\ell,k} \cap \Sc_{\ell,k'} \right| > 0 $. The parameters of the scheduling policy are chosen as $V=5,000$ and $A_{\rm max} = 100$. These choices have been empirically found to yield approximately ``near-optimal'' throughput rates, as explained below Theorem \ref{thm:performance_guarantees}. With larger values, no significant improvement was achieved, while reducing $V$ and $A_{\rm max}$ led to noticeably worse results.

In our simulations, we considered a total number of $K_{\rm tot} = 10,000$ users, which is representative of a dense area such as a sports stadium (see motivation in Section~\ref{intro}) to evaluate the described methods, i.e., a network where each frequency domain RB consists of $12$ subcarriers with subcarrier spacing of $60$ kHz, 
such that the bandwidth of each RB is $W_{\rm RB} = 720$ kHz. 
We divide the system bandwidth into $\lfloor W / (F W_{\rm RB}) \rfloor$ ``subchannels'' in frequency, 
each spanning $F$ RBs. The $K_{\rm tot}$ users are distributed among the different subchannels, such that on each subchannel $\Ktot = K_{\rm tot} \frac{F W_{\rm RB}}{W}$ users shall be served. 
Out of the $\Ktot$ users though, only a fraction $\Kact \approx LM/2$ users are scheduled simultaneously in order to operate the network at a reasonable user load.
Hence, by increasing $F$, the number of subchannels decreases and the number of users per subchannel
increases. This means that the users are scheduled less frequently, but when they are served, they transmit at higher rate (in bit/s) since the subchannel bandwidth also grows with $F$. In addition, larger $F$ yields larger frequency diversity, i.e., 
the CDF of the instantaneous mutual information ``concentrates'' due to the averaging in the frequency domain 
(see \eqref{eq:ul_mutual_inf} and \eqref{eq:dl_mutual_inf}). 
Since different values of $F$ yield different subchannel bandwidths, 
in order to compare the performance for different $F$, we need to consider the actual 
per-user throughput rates in ${\rm bit/s}$, obtained by multiplying the throughput rate in bit/s/Hz by the subchannel bandwidth, i.e., 
$\widetilde{\mu}_k := \bar{\mu}_k \times F W_{\rm RB}$. 
In the considered system,  the number of users per subchannel with $F = 1$ is given by $\Ktot(F=1) = K_{\rm tot} \frac{W_{\rm RB}}{W} = 120$. Then, for $F > 1$, the number of users per subchannel is given by $\Ktot(F) = F \times \Ktot(F=1)$.
 
In our results, we set $\tau_p = 20$, $\Kact = 70 \approx \frac{LM}{2}$ UEs per time slot, and $\widetilde{K}=80$ for the reassignment scheme. We identified this UL pilot dimension and user density regime empirically as a good choice for the considered network (a comparison of different $\tau_p$ and $\Kact$ is not shown here due to space limitations). We first evaluate the proposed schemes for a narrowband system with $F=1$ RB. Then, we consider the effect of higher frequency diversity $F = \{ 5, 10 \}$. For all simulations, we use the infinite buffer assumption with virtual arrivals and queues to achieve fairness among users. 


\subsection{Utility Optimization}

We consider HFS and PFS with the {\em fixed pilot} and {\em pilot reassignment} schemes, respectively. 
The proposed NUM-based scheduling approach is compared to a few ``baseline'' schemes. 
In particular, we have considered random selection, round-robin scheduling, and max-sum-rate scheduling.\footnote{Note that for the baseline schedulers, we use the pilot reassignment scheme. In case of conflicts, the pilot reassignment is repeated with a different user order (but same set of active users) until an assignment without conflicts is found.}
Random selection picks at each scheduling round $\Kact$ out of $\Ktot$ UEs per time slot, independent of the previous scheduling decisions. Round-robin scheduling sorts the UEs by their index and schedule them in lexicographic order, 
such that at scheduling slot $t = 1, 2, 3, \ldots$ the active user set is  $\{ t, t+1, \ldots, t+\Kact\}$ with indices repeated periodically modulo $\Ktot$.  The max-sum-rate scheduler selects in each time slot the $\Kact$ UEs 
to maximize  the sum expected service rate. This is equivalent to fixing the virtual queues 
in (\ref{user-selection}) such that $Q_k(t) = 1, \forall k, t$.
Fig.~\ref{fig:pfs_hfs_clusters_D} shows the per-user throughput CDF with $F = 1$ for PFS, HFS, and the 
three baseline schedulers. 
We notice that the max-sum-rate scheduler results in a very unfair throughput rate distribution, with a large number of 
UEs with zero throughput (see the jump at $\widetilde{\mu} = 0$ of the corresponding CDF). 
The PFS performs generally better than round robin and random scheduling. As expected, HFS equalizes the throughput rates across all UEs (the corresponding CDF is very close to a step function), and clearly yields a large improvement of the minimum rate with respect to PFS, while significantly reducing the maximum rates.
Also, Fig.~\ref{fig:pfs_hfs_clusters_D} compares the throughput CDF of PFS and HFS with fixed  pilots and  
pilot reassignment, for  $F=1$. We notice that for both HFS and PFS the degradation incurred by fixed pilots with respect to the 
more complex pilot reassignment scheme is very moderate. This indicates that although $\Kact$ is significantly smaller than $\Ktot$, allocating UL pilots to all users independently of the scheduling decision, and performing 
the WSRM under the proposed conflict graph constraint, is indeed an attractive approach.

\begin{figure}[t!]
	\centerline{\includegraphics[width=.7\linewidth]{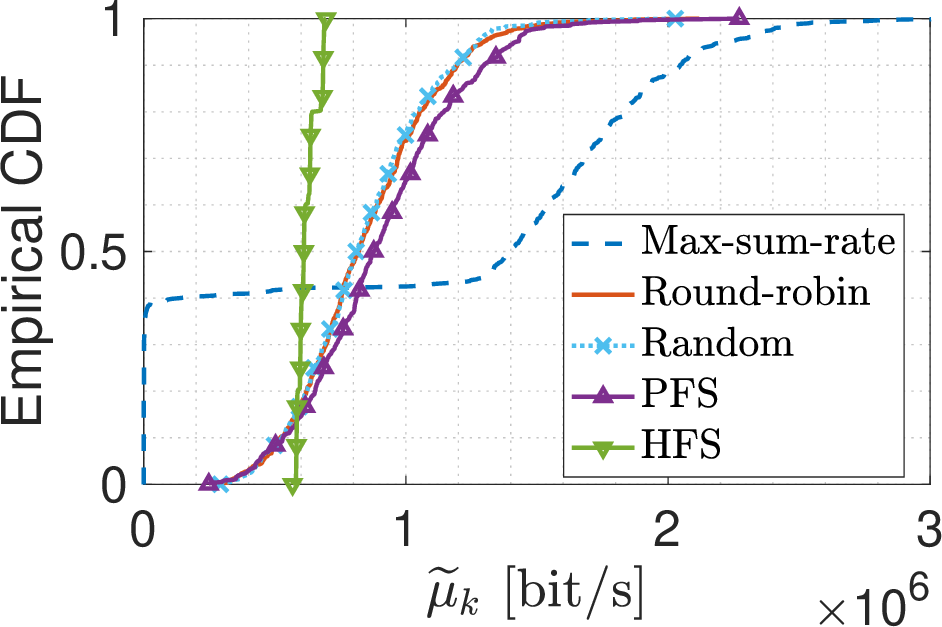}}
	\centerline{\includegraphics[width=.49\linewidth]{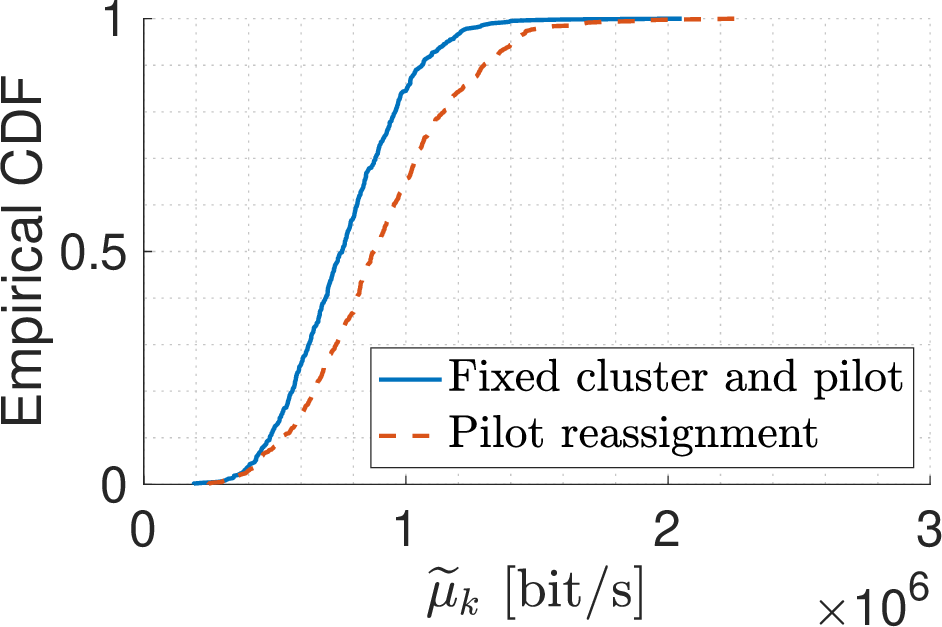} \includegraphics[width=.49\linewidth]{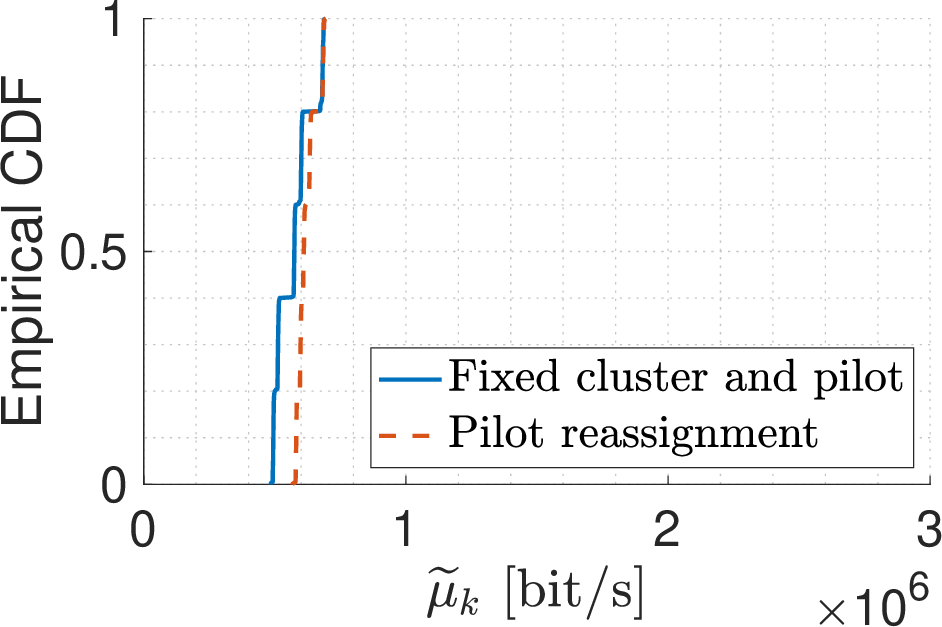}}
	\caption{The empirical CDF of the user throughput for PFS, HFS with pilot reassignment and the baseline schedulers (top). The empirical CDF of the user throughput for PFS (bottom left) and HFS (bottom right) with fixed pilots and pilot reassignment.}
	\label{fig:pfs_hfs_clusters_D}
	\vspace{-.4cm}
\end{figure}

\subsection{Effect of the frequency diversity}

Here we wish to assess the effect of increasing the frequency diversity by forming wider band subchannels with 
$F=\{5, 10\}$ RBs.  We compare all systems under pilot reassignment. 
Fig.~\ref{fig:empirical_cdf_narrow_wideband} shows the hardening of the instantaneous mutual information 
with increased frequency diversity order. In fact, the empirical CDF of the mutual information is less and less spread as $F$ increases. This allows a more aggressive instantaneous rate allocation in the active slots. As a result (see Fig.~\ref{fig:empirical_cdf_narrow_wideband}), the user throughput rates in a system with $F=\{ 5, 10 \}$ can be significantly increased compared to $F=1$ for both PFS and HFS. 
This is also evidenced in Fig.~\ref{fig:obj_function_narrow_wideband}, showing  the geometric mean of the user throughputs under PFS,  and the minimum user throughput under HFS. Notice that the former is directly related to the PFS objective function, since obviously $\left( \prod_{k=1}^{\Ktot(F)} \widetilde{\mu}_k \right)^\frac{1}{\Ktot(F)} = \exp\left( \frac{1}{\Ktot(F)} \left( \sum_{k=1}^{\Ktot(F)} \log \widetilde{\mu}_k \right) \right)$. 

We also observe that the improvement from $F=1$ to $F=5$ is quite significant, while 
further increasing the frequency diversity to $F=10$ yields a smaller performance gain, especially for HFS. This indicates 
a sort of saturation of the benefit provided by frequency diversity. As a matter of fact, since
the number of users per subchannel $\Ktot(F)$ increases linearly with $F$, the scheduler must solve a larger 
WSRM problem for larger $F$. Hence, it is advisable to choose a moderate value of $F$ that yields good frequency diversity gain
but not a too complex scheduler.

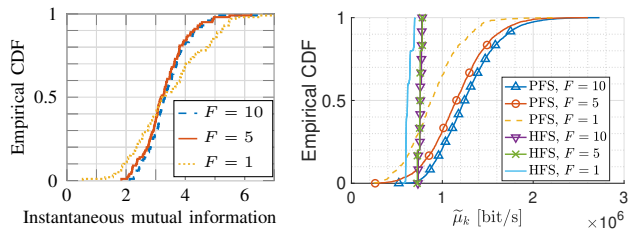
\begin{figure}[t!]
	\centerline{
	\begin{tikzpicture}[define rgb/.code={\definecolor{mycolor}{RGB}{#1}},
		rgb color/.style={define rgb={#1},mycolor}]
		\begin{axis}[
			enlarge y limits=false,clip=false,
			enlarge x limits=false,clip=true,
			no markers,
			width=.49\linewidth, 
			height=3.8cm, 
			grid=major, 
			grid style={gray!30}, 
			xlabel={\scriptsize Instantaneous mutual information},
			x label style={at={(axis description cs: 0.4, -.29)},rotate=0,anchor=south},
			ylabel={\scriptsize Empirical CDF},
			y label style={at={(axis description cs: -0.16, .5)},rotate=0,anchor=south},
			label style={inner sep=2pt}, 
			legend style={at={(0.5,-0.2)},anchor=north}, 
			yticklabels={},
			extra y ticks={0, .1, .2, .3, .4, 0.5, .6, .7, .8, .9, 1},
			extra y tick labels={0,,,,,0.5,,,,,1},
			xticklabels={},
			extra x ticks={0, 1, 2, 3, 4, 5, 6, 7},
			extra x tick labels={0, , 2, , 4, , 6, },
			x tick label style={yshift=.7,rotate=0,anchor=north}, 
			tick label style={font=\scriptsize},
			xmin=0,
			xmax=7,
			ymin=0,
			ymax=1,
			legend pos=south east,
			legend style={row sep=-0.0cm, inner ysep = 0cm, inner xsep = 0.05cm, font=\fontsize{6}{12}, at={(0.995,0.02)}, anchor=south east},
			legend cell align={left},
			legend style={kurze Legende}]
			\addplot[rgb color={0, 114, 189}, line width=0.8pt, dashed]
			table[blue,col sep=comma] {csvs/data_mut_info_wideband_F10.csv}; 			
			\addlegendentry{{ $F=10$}};
			\addplot[rgb color={217, 83, 25}, line width=0.8pt]
			table[blue,col sep=comma] {csvs/data_mut_info_wideband_F5.csv}; 			
			\addlegendentry{{ $F=5$}};
			\addplot[rgb color={237, 177, 32}, line width=0.8pt, densely dotted]
			table[blue,col sep=comma] {csvs/data_mut_info_wideband_F1.csv}; 			
			\addlegendentry{{ $F=1$}};
			\begin{scope}[scale=1.1, transform shape]
				\draw (5.25, -.22) node[left,black]{$\Ic_k\left( \vvv_k { , } \HH \right)$ \si{[bit/s/Hz]}};
			\end{scope}
		\end{axis}
	\end{tikzpicture}
	\hspace{.0cm} 
	\includegraphics[width=.49\linewidth]{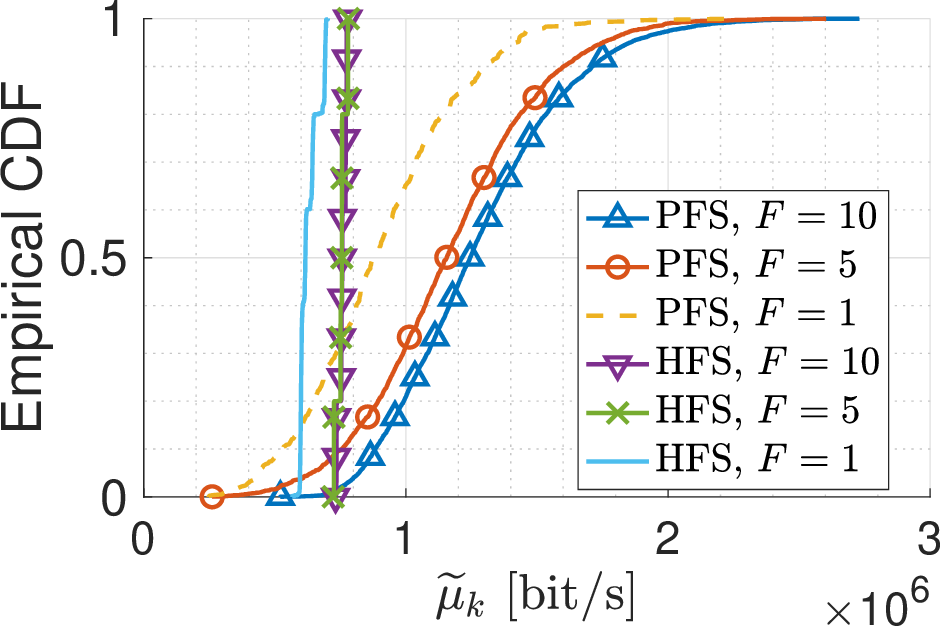} 
	}
	\vspace{-.1cm}
	\caption{The empirical CDF of $\Ic_k\left( \vvv_k, \HH \right)$ from $N=100$ samples for a given typical user $k$  in the center of the coverage area (left). The empirical CDF of the user throughput for PFS/HFS with $F = \{1, 5, 10\}$ using the pilot reassignment scheme (right).}
	\label{fig:empirical_cdf_narrow_wideband}
	\vspace{-.5cm}
\end{figure}

\begin{figure}[t!]
	\centerline{\includegraphics[width=.49\linewidth]{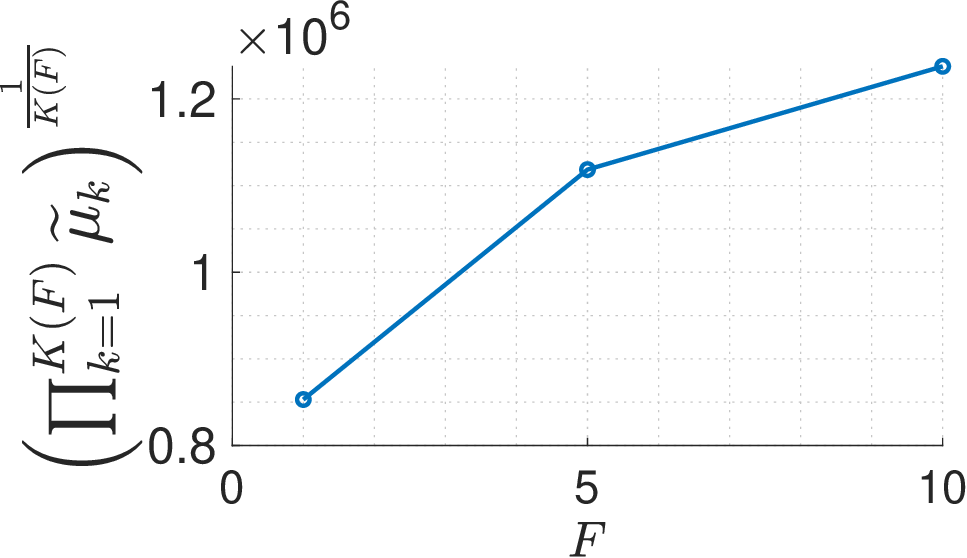} 
	\hspace{.0cm} 
	\includegraphics[width=.49\linewidth]{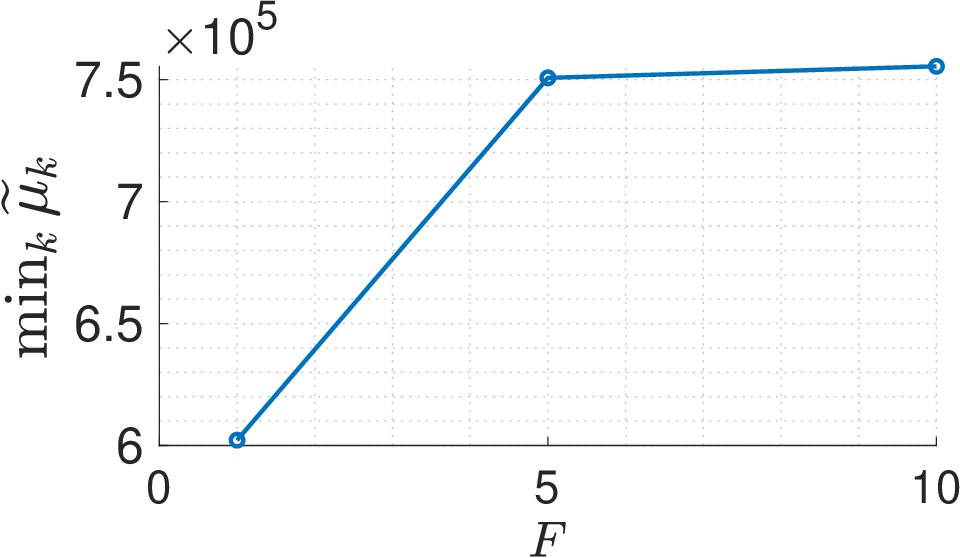}}
	\vspace{-.1cm}
	\caption{The geometric mean of the user throughput under PFS (left) 
	and the minimum throughput under HFS (right).}
	\label{fig:obj_function_narrow_wideband}
	\vspace{-.5cm}
\end{figure}

\subsection{Downlink Scheduling}

Although in the paper we have mainly considered the UL for the sake of exposition, the same approach can be applied to the DL.
Fig.~\ref{fig:downlink_distributed_pfs} shows the UL/DL results for PFS and HFS, respectively, 
in the case of $F=1$ and pilot reassignment, where the DL precoding vectors are identical to the UL detection vectors 
and uniform power allocation over all DL data streams is used, such that the total transmit power for UL and DL is identical, as explained previously. We notice that a similar user throughput is achieved in both cases, confirming the approximate duality results of  our previous work \cite{kddi_uldl_precoding}. 
Note that in the case of practical relevance of an imbalance in UL and DL traffic demands, 
a different number of time slots can be allocated to the UL and DL to meet the respective demand. 
Therefore, while our results show {\em balanced} UL and DL throughput distributions, it is clear that by unequal slot allocation 
one can adapt UL and DL to the actual traffic demand.

\section{Conclusions}

In this work, we considered a user-centric cell-free massive MIMO network with a total number of users 
much larger than the optimal user load. In order to serve all users in the network with a fair distribution of throughput rates, 
we proposed a dynamic scheduling scheme based on Network Utility Maximization via the 
Lyapunov DPP approach. While the approach is quite standard, we have addressed several issues that are specific of the system at hand and represent the main novelty of this work. 
In particular, we considered the problem of pilot and cluster assignment, which can be fixed for all users, or dependent on the scheduling decision (reassignment). 
The key component of the dynamic scheduler is a novel conflict-graph constrained 
WSRM problem to be solved at each scheduling round, in the form of an integer linear program.  
Also, we considered the problem of instantaneous rate allocation in the information outage regime, based on the empirical CDF of the instantaneous mutual information that each UE can accumulate on a window of past time slots. This is very different from the standard works on cell-free massive MIMO, that consider ergodic rates and users 
continuously active on a (virtually infinite) sequence of fading states. 
In our case, block by block coding/decoding is dictated by the fact that users are scheduled on generally non-consecutive slots, and coding across slots would result in excessive decoding delay, incompatible with the low latency requirement of 5G systems. 
The use of information outage rates also illuminated the role of frequency diversity. For a given total number of users in the system and a total system bandwidth, allocating wider subchannels yields larger frequency diversity order, but also more users per subchannel. This means that each user is active for a smaller fraction of time, but when active it transmits at higher rates
(in bit/s). We have verified that a moderate amount of frequency  diversity is indeed beneficial for the user throughput. However, this benefit tends to saturate and since larger subband channels involve a higher complexity in the scheduler (namely, a larger size of the integer program to solve the WSRM problem at each scheduling round), it is advisable to 
carefully dimension the subchannel bandwidth in order to achieve a good tradeoff between throughput performance and scheduler complexity. 
\begin{figure}[t!]
	\centerline{\includegraphics[width=.49\linewidth]{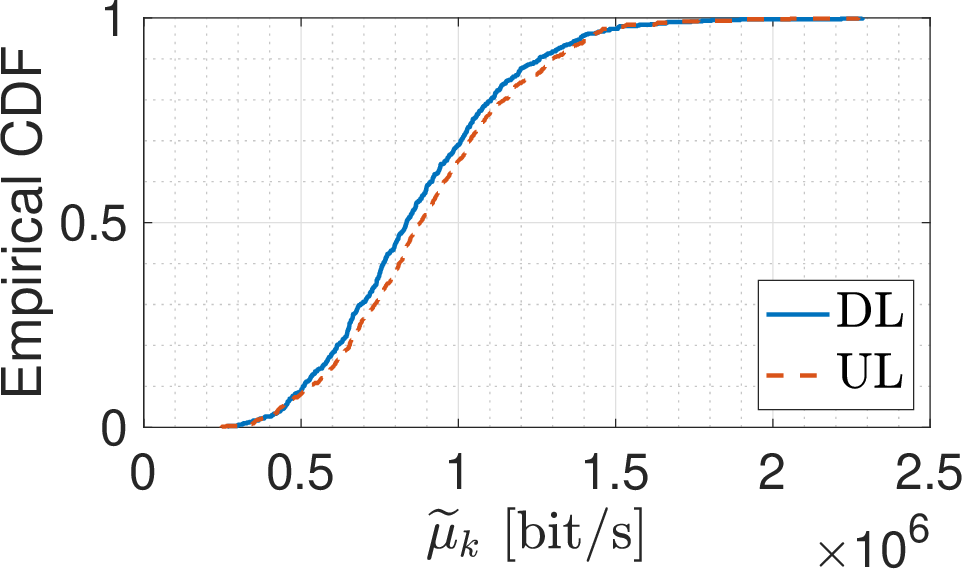}
		\includegraphics[width=.49\linewidth]{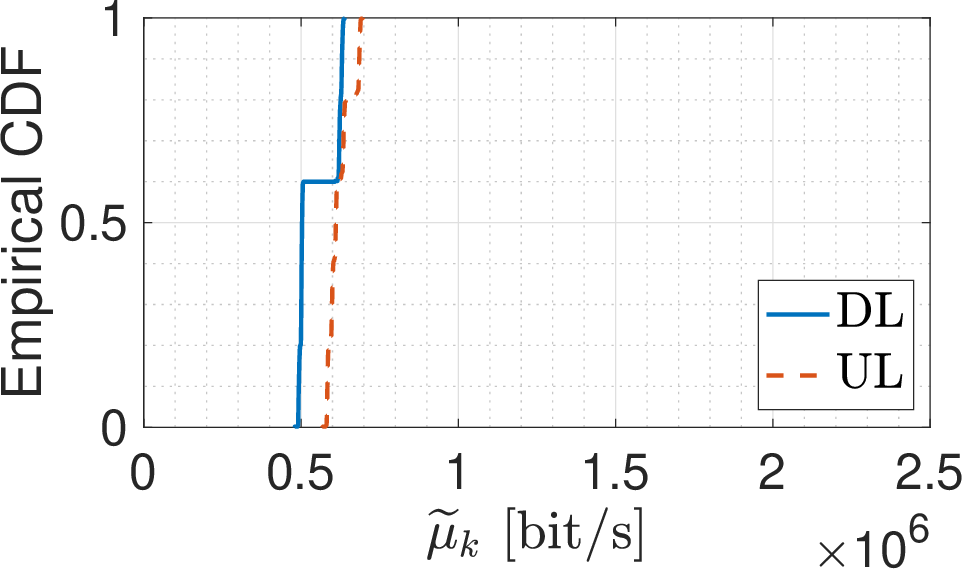} }
	\vspace{-.2cm}
	\caption{The empirical CDF of the DL and UL user throughput for PFS (left) and HFS (right). In both cases $F=1$.}
	\vspace{-.5cm}
	\label{fig:downlink_distributed_pfs}
\end{figure}

Overall, we have demonstrated the effectiveness of the proposed approach by considering a system with 
10,000 users in $0.2 \times 0.2$ km$^2$, 
a total bandwidth of 60 MHz and 200 infrastructure antennas (20 RUs with 10 antennas each). 
Under PFS with $F = 5$ the system achieves a geometric mean throughput rate per user of approximately $1.1$ Mb/s (see Fig.~\ref{fig:obj_function_narrow_wideband}). This corresponds to a total rate (sum over all users) of 11 Gb/s over 60 MHz, 
i.e., a total spectral efficiency of 183.3 bit/s/Hz per $0.2 \times 0.2$ km$^2$ or, equivalently, 
4582.5 bit/s/Hz per km$^2$. We conclude that such system is capable of providing an excellent mean throughput rate with 
fairness among a large population of users, and it is therefore a very attractive solution to serve extremely dense localized areas, such as sport stadiums, train stations, shopping malls, and similar. 

\appendices

\section{Proof of Theorem~\ref{thm:performance_guarantees}} \label{app:proofs}
For the proof of Theorem~\ref{thm:performance_guarantees}, we use the Lyapunov drift approach. We let $\Qm = \left[ Q_1 \dots Q_\Ktot \right]^\transp$, and the Lyapunov function defined on $\RR_{+}^\Ktot$ is given by $\Lc(\Qm) = \frac{1}{2} \sum_{k=1}^{\Ktot} Q_k^2$ with the corresponding one-step Lyapunov drift
 \begin{gather}
 	\Delta(\Qm(t)) = \EE [ \Lc(\Qm (t+1) ) - \Lc(\Qm (t) ) | \Qm(t) ] \label{eq:def_lyapunov_drift}.
 \end{gather}
Further, we use (\ref{Qk-evo-virtual}) and write
\begin{align}
	&Q_k(t+1)^2 \nonumber \\
	&= \left( \max\left\{ Q_k(t) - \mu_k(t), 0 \right\} + A_k(t) \right)^2 \nonumber \\
	& \leq [Q_k(t) - \mu_k(t)]^2 + A_k^2(t) \nonumber \\ & \hspace{1cm} + 2 A_k(t) \max\left\{ Q_k(t) - \mu_k(t), 0 \right\} \nonumber \\
	& = Q_k^2(t) + \mu_k^2(t) - 2Q_k(t)\mu_k(t) + A_k^2(t) \nonumber \\ & \hspace{1cm} + 2 A_k(t) \max\left\{ Q_k(t) - \mu_k(t), 0 \right\} \nonumber \\		
	& \leq Q_k^2(t) + \mu_k^2(t) + A_k^2(t) - 2 Q_k(t) \left[ \mu_k(t) - A_k(t) \right]. \nonumber
\end{align}
Summing with respect to $k$ and applying the conditional expectation $\EE[\cdot | \Qm(t)]$, we arrive at
\begin{align}
	\Delta \Qm(t) & \leq \frac{1}{2} \sum_{k=1}^{\Ktot} \EE\left[ \mu^2_k(t) + A^2_k(t) | \Qm(t) \right] \nonumber\\ 
	& \hspace{1cm} - \sum_{k=1}^{\Ktot} Q_k(t) \EE\left[ \mu_k(t) - A_k(t) | \Qm(t) \right]. \label{eq:DeltaQ_leq_R_A}
\end{align}
In addition, we need the following Lemma to prove Theorem~\ref{thm:performance_guarantees}.
\begin{lemma} \label{lemma:queues_R}
	Let the service rates $\left\{ \mu_k(t) \right\}$ be obtained by the scheduler $\gammadyn$. Then, for any $\bar{\muv} \in \Rs$, we have
	$\sum_{k=1}^{\Ktot} Q_k(t) \EE \left[ \mu_k(t) | \Qm(t) \right] \geq \sum_{k=1}^{\Ktot} Q_k(t) \bar{\mu}_k \label{eq:queues_Rbar}$.
\end{lemma}
\ \ \ \textit{Proof:} Notice that $\Rs$ is a convex compact region in $\RR_{+}^\Ktot$. For any fixed non--negative weight vector $\Qm$, the maximum of the linear function $\sum_{k=1}^{\Ktot} Q_k r_k$, where $\rv \in \Rs$, is achieved by some $\gamma \in \Gamma$. 
We let $\Omegam(t)$ denote the mutual information statistics available in time slot $t$ used to compute (\ref{inst-rate-allocation}).
Hence, for any $\bar{\muv} \in \Rs$ and weight vector $\Qm(t)$, there exists $\gamma \in \Gamma$ such that  
\begin{align}
 	&\sum_{k=1}^{\Ktot} Q_k(t) \bar{\mu}_k \nonumber \\
 	&= \sum_{k=1}^{\Ktot} Q_k(t) \EE \left[ \mu_k\left( \HH(t), \gamma( \Omegam(t) ) \right) \right] \nonumber \\
 	&= \sum_{k=1}^{\Ktot} Q_k(t) \EE \Bigl[ \EE \left[ \mu_k\left( \HH(t), \gamma( \Omegam(t) ) | \Omegam(t), \gamma \right) \right] \Bigr] \nonumber \\
 	&\leq \EE \left[  \max\limits_{\Ac(t) } \sum_{k=1}^{\Ktot} Q_k(t) \EE \left[  \mu_k\left( \HH(t), \Ac(t) \right) | \Omegam(t) \right] \bigg| \Qm(t)  \right] \nonumber \\
 	&= \sum_{k=1}^{\Ktot} Q_k(t) \EE \Bigl[  \EE \left[  \mu_k\left( \HH(t), \gammadyn ( \Omegam(t) ) \right) | \Omegam(t) \right] \Big| \Qm(t) \Bigr] .
 \end{align}
 Since we assumed that the service rates $\mu_k(t)$ are obtained by applying the policy $\gammadyn$, then by definition  $\EE \Bigl[  \EE \left[  \mu_k\left( \HH(t), \gammadyn ( \Omegam(t) ) \right) | \Omegam(t) \right] \Big| \Qm(t) \Bigr] = \EE \left[ \mu_k(t) | \Qm(t) \right]$, and the Lemma is proved. \hfill $\blacksquare$

\textbf{Proof of Theorem~\ref{thm:performance_guarantees}.} We define $\overline{\Am} (t) = \frac{1}{t} \sum_{\tau = 0}^{t-1} \EE [\Am(\tau)]$ and $\overline{\muv} (t) = \frac{1}{t} \sum_{\tau = 0}^{t-1} \EE [\muv(\tau)]$, where $\Am(\tau) = \left[ A_1(\tau) \dots A_\Ktot(\tau) \right]^\transp$ and $\muv(\tau)$ are the virtual arrival process and the service rate vector induced by $\gammadyn$. From \cite{shirani2010mimo}, we know that
 \begin{align}
 	\frac{1}{t} \sum_{\tau=0}^{t-1} \EE [ \Am(\tau) ] \leq \frac{1}{t} \sum_{\tau=0}^{t-1} \EE [ \muv(\tau) ] + \frac{ \EE [ \Qm(t) ] }{ t } \label{eq:A_leq_R_plus_Q} .
 \end{align}
We will use the following preliminary fact, whose proof uses (\ref{eq:A_leq_R_plus_Q}) and the fact that strong stability and uniformly bounded arrival processes imply mean-rate stability (i.e., $\EE [ \Qm(t) /t \rightarrow \zerov ]$) \cite{georgiadis2006resource}.
\begin{fact} \label{fact2:limit_A_R}
	We assume that the queues $\Qm(t)$ are strongly stable and that there is a finite upper bound $A_{\rm max}$ on arrivals for all $t$. If $g(\cdot)$ is a continuous and componentwise non-decreasing function, then
	\begin{align}
		\liminf_{t\rightarrow\infty} g( \overline{\Am}(t) ) &\leq \liminf_{t\rightarrow\infty} g( \overline{\muv}(t) )  \label{eq:liminf_A_leq_R} , \\
		\limsup_{t\rightarrow\infty} g( \overline{\Am}(t) ) &\leq g( \bar{\muv}^\star (A_{\rm max} ) ) \label{eq:limsup_A_leq_R} .
	\end{align}
\end{fact}
\hfill $\lozenge$

We further observe that $\EE[\mu_{k^\star}^2(t)] \leq 
	\EE \left[ \left( \frac{1}{F} \sum_{f=1}^{F} \log \left( 1 + \frac{ | \hh_{k^\star}(t, f) |^2 }{\SNR^{-1}} \right) \right)^2 \right]  \label{eq:max_mu} $,
where the latter is the maximum achievable instantaneous rate of UE $k^\star$, given by $\underset{ k \in [\Ktot] }{\arg\max} \ \frac{1}{F} \sum\limits_{f=1}^{F} \log \left( 1 + \frac{ | \hh_k(t, f) |^2 }{\SNR^{-1}} \right)$,
under perfect CSI and as if it was alone in the system.
It follows that
\begin{align}
	&\frac{1}{2} \sum\limits_{k=1}^{\Ktot} \EE\left[ \mu^2_k(t) + A^2_k(t) | \Qm(t) \right] \nonumber \\ 
	& \leq \frac{\Ktot}{2} \left( A_{\rm max}^2 + \EE \left[ \left( \frac{1}{F} \sum\limits_{f=1}^{F} \log \left( 1 + \frac{ | \hh_{k^\star}(t, f) |^2 }{\SNR^{-1}} \right) \right)^2 \right] \right) \nonumber \\ 
	& \defines C < \infty . \label{eq:def_C}
\end{align}

From (\ref{eq:DeltaQ_leq_R_A}), (\ref{eq:def_C}) and Lemma~\ref{lemma:queues_R}, we know that
\begin{align}
	\Delta(\Qm(t)) \leq C - \sum_{k=1}^{\Ktot} Q_k(t) \bar{\mu}_k + \sum_{k=1}^{K} Q_k(t) \EE \left[A_k(t) | \Qm(t) \right] \label{eq:thm2_lyapunov_drift_start} ,
\end{align}
where $\Delta(\Qm(t))$ is the Lyapunov drift defined in (\ref{eq:def_lyapunov_drift}), $C$ is given in (\ref{eq:def_C}) and $\bar{\muv} = \left[\bar{\mu}_1, \dots, \bar{\mu}_\Ktot \right]^\transp$ is any rate vector in $\Rs$. Following the technique in \cite{georgiadis2006resource, neely2008fairness}, we subtract a term related to the utility function from both sides of (\ref{eq:thm2_lyapunov_drift_start}), which yields
\begin{align}
	& \Delta( \Qm(t) ) - V \EE \left[ g(\Am(t)) | \Qm(t) \right] \nonumber \\
	& \leq C - \sum_{k=1}^{\Ktot} Q_k(t) \bar{\mu}_k  \nonumber \\ 
	& \hspace{2cm} + \EE \left[ \sum_{k=1}^{K} Q_k(t) A_k(t) - V g(\Am(t)) \Big| \Qm(t) \right] .
\end{align} 
We note that $\gammadyn$ is defined in (\ref{arrival-rate}) such that it minimizes the right hand side over all vectors $\av$ that satisfy $0 \leq a_k \leq A_{\rm max}$ for all $k$. Now, let $\zv$ be any vector in $\Rs$ that satisfies $0 \leq z_k \leq A_{\rm max}$ for all $k$. Then 
\begin{align}
	&\Delta( \Qm(t) ) - V \EE \left[ g(\Am(t)) | \Qm(t) \right] \nonumber \\
	& \leq C - \sum_{k=1}^{\Ktot} Q_k(t) \bar{\mu}_k + \sum_{k=1}^{K} Q_k(t) z_k(t) - V g(\zv) .
\end{align}
Taking expectations of both sides of the above inequality and using the law of iterated expectations yields
\begin{align}
	& \EE[ \Lc (\Qm(t+1) ) ] - \EE[ \Lc (\Qm(t)) ] - V\EE[ g(\Am(t)) ] \nonumber \\
	& \leq C - \sum_{k=1}^{\Ktot} \EE[Q_k(t)] (\bar{\mu}_k - z_k) -Vg(\zv) .
\end{align}
We assume $\Qm(0) = \zerov$ for simplicity. The above inequality holds for all $t$. Summing over $\tau \in \{ 0, \dots, t-1 \}$, dividing by $t$, rearranging terms, and using the non-negativity of $\Lc(\cdot)$ we have
\begin{align}
	\frac{1}{t} \sum_{\tau = 0}^{t - 1}\sum_{k=1}^{K} \EE [Q_k(\tau)] (\bar{\mu}_k - z_k) \leq C + Vg(\overline{\Am}(t)) - Vg(\zv) \label{eq:queues_jensens_ineq} , 
\end{align}
where we used Jensen's inequality in the concave function $g(\cdot)$. 
The above inequality holds for all $t$, all $\bar{\muv} \in \Rs$, and all $\zv \in \Rs$ such that $0 \leq z_k \leq A_{\rm max}$ for all $k$. Parts (a) and (b) of Theorem~\ref{thm:performance_guarantees} are proven by plugging different values into (\ref{eq:queues_jensens_ineq}). We first prove part (b).

\textit{Proof of part (b).} Take any point $\tilde{\zv} \in \Rs$ such that $\tilde{\zv} = \left[ \tilde{z}_1 \dots \tilde{z}_\Ktot \right]^\transp $ and $0 \leq \tilde{z}_k \leq A_{\rm max}$ for all $k$. Choose $\bar{\muv} = \tilde{\zv}$ and $\zv = \beta \tilde{\zv}$, for any $\beta \in [0, 1]$. Then, from (\ref{eq:queues_jensens_ineq}), we have
\begin{gather}
	\frac{1}{t} \sum_{\tau = 0}^{t-1} \sum_{k=1}^{\Ktot} \tilde{z}_k \EE[Q_k(\tau)] \leq \frac{C + V g(\overline{\Am}(t)) - V g(\beta\tilde{\zv}) }{1 - \beta} . \label{eq:queues_jensens_ineq_x_z}
\end{gather}
Now, we first prove that the queues are strongly stable and then, using Fact~\ref{fact2:limit_A_R}, we obtain part (b) of Theorem~\ref{thm:performance_guarantees}. Notice that $g( \overline{\Am}(t) ) \leq g(\Am_{\rm max})$, where $\Am_{\rm max}$ is a vector with each entry equal to $A_{\rm max}$. Using this bound in (\ref{eq:queues_jensens_ineq_x_z}) and taking a $\limsup$ yields
\begin{gather}
	\limsup_{t\rightarrow\infty} \frac{1}{t} \sum_{\tau=0}^{t-1} \sum_{k = 1}^{\Ktot} \tilde{z}_k \EE[ Q_k(t) ] \leq \frac{C + V g(\Am_{\rm max}) - V g(\beta\tilde{\zv}) }{1 - \beta} . \label{eq:queues_jensens_ineq_Amax}
\end{gather}
By assumption, there exists at least one point $\rv \in \Rs$ that has all positive entries and such that $g(\frac{\rv}{2}) > - \infty$. Choosing $\beta = \frac{1}{2}$ and $\tilde{\zv} = \rv$, it follows that the right-hand site of (\ref{eq:queues_jensens_ineq_Amax}) is finite and hence all queues are strongly stable.

Because of strong stability and since the arrival processes are bounded by $A_{\rm max} < \infty$, we can apply inequality (\ref{eq:limsup_A_leq_R}) of Fact~\ref{fact2:limit_A_R} to the right-hand site of (\ref{eq:queues_jensens_ineq_x_z}) after taking a $\limsup$ and obtain the result of part (b).

\textit{Proof of part (a).} We obtain
$g( \overline{\Am} (t) ) \geq g(  \bar{\muv}^\star(A_{\rm max})  ) - \frac{C}{V}$ by plugging $\bar{\muv} = \zv = \bar{\muv}^\star (A_{\rm max})$ into (\ref{eq:queues_jensens_ineq}).
By taking $\liminf$ and using (\ref{eq:liminf_A_leq_R}) of Fact~\ref{fact2:limit_A_R}, we obtain the result of (a).

 \bibliography{IEEEabrv,scheduling-journal}

\begin{thebibliography}{10}
\providecommand{\url}[1]{#1}
\csname url@samestyle\endcsname
\providecommand{\newblock}{\relax}
\providecommand{\bibinfo}[2]{#2}
\providecommand{\BIBentrySTDinterwordspacing}{\spaceskip=0pt\relax}
\providecommand{\BIBentryALTinterwordstretchfactor}{4}
\providecommand{\BIBentryALTinterwordspacing}{\spaceskip=\fontdimen2\font plus
\BIBentryALTinterwordstretchfactor\fontdimen3\font minus
  \fontdimen4\font\relax}
\providecommand{\BIBforeignlanguage}[2]{{%
\expandafter\ifx\csname l@#1\endcsname\relax
\typeout{** WARNING: IEEEtran.bst: No hyphenation pattern has been}%
\typeout{** loaded for the language `#1'. Using the pattern for}%
\typeout{** the default language instead.}%
\else
\language=\csname l@#1\endcsname
\fi
#2}}
\providecommand{\BIBdecl}{\relax}
\BIBdecl

\bibitem{Caire-Shamai-TIT03}
G.~Caire and S.~{Shamai (Shitz)}, ``{On the achievable throughput of a
  multiantenna Gaussian broadcast channel},'' \emph{{IEEE} Trans. on Inform.
  Theory}, vol.~49, no.~7, pp. 1691--1706, July 2003.

\bibitem{Viswanath-Tse-TIT03}
P.~Viswanath and D.~N.~C. Tse, ``{Sum capacity of the vector Gaussian broadcast
  channel and uplink-downlink duality},'' \emph{{IEEE} Trans. on Inform.
  Theory}, vol.~49, no.~8, pp. 1912--1921, Aug. 2003.

\bibitem{Weingarten-Steinberg-Shamai-TIT06}
H.~Weingarten, Y.~Steinberg, and S.~{Shamai (Shitz)}, ``{The capacity region of
  the Gaussian multiple-input multiple-output broadcast channel},''
  \emph{{IEEE} Trans. on Inform. Theory}, vol.~52, no.~9, pp. 3936--3964, Sept.
  2006.

\bibitem{Caire-Jindal-Kobayashi-Ravindran-TIT10}
G.~Caire, N.~Jindal, M.~Kobayashi, and N.~Ravindran, ``{Multiuser MIMO
  achievable rates with downlink training and channel state feedback},''
  \emph{{IEEE} Trans. on Inform. Theory}, vol.~56, no.~6, pp. 2845--2866, June
  2010.

\bibitem{3gpp38211}
3GPP, ``{NR; Physical channels and modulation},'' 3GPP Tech. Spec. 38.211, 04
  2022, {Version 17.1.0}.

\bibitem{Larsson-book}
T.~L. Marzetta, E.~G. Larsson, H.~Yang, and H.~Q. Ngo, \emph{Fundamentals of
  Massive MIMO}.\hskip 1em plus 0.5em minus 0.4em\relax Cambridge University
  Press, 2016.

\bibitem{cell-free-fnt}
{\"O}.~T. Demir, E.~Bj{\"o}rnson, L.~Sanguinetti \emph{et~al.}, ``{Foundations
  of User-Centric Cell-Free Massive MIMO},'' \emph{Foundations and
  Trends{\textregistered} in Signal Processing}, vol.~14, no. 3-4, pp.
  162--472, 2021.

\bibitem{khorov2018tutorial}
E.~Khorov, A.~Kiryanov, A.~Lyakhov, and G.~Bianchi, ``{A tutorial on IEEE
  802.11 ax high efficiency WLANs},'' \emph{IEEE Communications Surveys \&
  Tutorials}, vol.~21, no.~1, pp. 197--216, Sept. 2018.

\bibitem{qu2019survey}
Q.~Qu, B.~Li, M.~Yang, Z.~Yan, A.~Yang, D.-J. Deng, and K.-C. Chen, ``{Survey
  and performance evaluation of the upcoming next generation WLANs
  standard-IEEE 802.11 ax},'' \emph{Mobile Networks and Applications}, vol.~24,
  no.~5, pp. 1461--1474, Oct. 2019.

\bibitem{marzetta2010noncooperative}
T.~L. Marzetta, ``Noncooperative cellular wireless with unlimited numbers of
  base station antennas,'' \emph{{IEEE Trans. on Wireless Comm.}}, vol.~9,
  no.~11, pp. 3590--3600, 2010.

\bibitem{hoydis2013massive}
J.~Hoydis, S.~Ten~Brink, and M.~Debbah, ``{Massive MIMO in the UL/DL of
  cellular networks: How many antennas do we need?}'' \emph{IEEE Journal on
  selected Areas in Communications}, vol.~31, no.~2, pp. 160--171, 2013.

\bibitem{huh2012achieving}
H.~Huh, G.~Caire, H.~C. Papadopoulos, and S.~A. Ramprashad, ``{Achieving
  "Massive MIMO" Spectral Efficiency with a Not-so-Large Number of Antennas},''
  \emph{IEEE Trans. on Wireless Commun.}, vol.~11, no.~9, pp. 3226--3239, 2012.

\bibitem{ngo2017}
H.~Q. Ngo, A.~Ashikhmin, H.~Yang, E.~G. Larsson, and T.~L. Marzetta,
  ``{Cell-Free Massive MIMO Versus Small Cells},'' \emph{IEEE Trans. on
  Wireless Commun.}, vol.~16, no.~3, pp. 1834--1850, 2017.

\bibitem{nay2017}
E.~Nayebi, A.~Ashikhmin, T.~L. Marzetta, H.~Yang, and B.~D. Rao, ``{Precoding
  and Power Optimization in Cell-Free Massive MIMO Systems},'' \emph{IEEE
  Trans. on Wireless Commun.}, vol.~16, no.~7, pp. 4445--4459, 2017.

\bibitem{bjo2020}
E.~Bj{\"o}rnson and L.~Sanguinetti, ``{Scalable Cell-Free Massive MIMO
  Systems},'' \emph{{IEEE Trans. on Comm.}}, vol.~68, no.~7, pp. 4247--4261,
  2020.

\bibitem{miretti2022}
L.~Miretti, E.~Bj{\"o}rnson, and D.~Gesbert, ``{Team MMSE Precoding With
  Applications to Cell-Free Massive MIMO},'' \emph{IEEE Trans. on Wireless
  Commun.}, vol.~21, no.~8, pp. 6242--6255, 2022.

\bibitem{goettsch2022}
F.~G{\"o}ttsch, N.~Osawa, T.~Ohseki, K.~Yamazaki, and G.~Caire,
  ``{Subspace-Based Pilot Decontamination in User-Centric Scalable Cell-Free
  Wireless Networks},'' \emph{IEEE Trans. on Wireless Commun.}, vol.~22, no.~6,
  pp. 4117--4131, 2023.

\bibitem{3gpp38300}
3GPP, ``{NR; NR and NG-RAN Overall description; Stage-2},'' 3GPP Tech. Spec.
  38.300, 05 2022, {Version 17.0.0}.

\bibitem{chen2022}
S.~Chen, J.~Zhang, E.~Bj{\"o}rnson, and B.~Ai, ``{Improving Fairness for
  Cell-Free Massive MIMO Through Interference-Aware Massive Access},''
  \emph{IEEE Transactions on Vehicular Technology}, pp. 1--6, 2022.

\bibitem{itu2017requirements}
ITU-R, ``{Minimum requirements related to technical performance for IMT-2020
  radio interface(s)},'' Report ITU-R M.2410-0, 11 2017.

\bibitem{7529226}
G.~Durisi, T.~Koch, and P.~Popovski, ``{Toward Massive, Ultrareliable, and
  Low-Latency Wireless Communication With Short Packets},'' \emph{Proceedings
  of the IEEE}, vol. 104, no.~9, pp. 1711--1726, 2016.

\bibitem{gottsch2022fairness}
F.~G{\"o}ttsch, N.~Osawa, T.~Ohseki, Y.~Amano, I.~Kanno, K.~Yamazaki, and
  G.~Caire, ``{Fairness Scheduling in Dense User-Centric Cell-Free Massive MIMO
  Networks},'' in \emph{2022 56th Asilomar Conference on Signals, Systems, and
  Computers}, 2022, pp. 733--737.

\bibitem{7421132}
A.~Forenza, S.~Perlman, F.~Saibi, M.~Di~Dio, R.~van~der Laan, and G.~Caire,
  ``{Achieving large multiplexing gain in distributed antenna systems via
  cooperation with pCell technology},'' in \emph{2015 49th Asilomar Conference
  on Signals, Systems and Computers}, 2015, pp. 286--293.

\bibitem{biglieri1998fading}
E.~Biglieri, J.~Proakis, and S.~Shamai, ``{Fading channels:
  Information-theoretic and communications aspects},'' \emph{IEEE transactions
  on information theory}, vol.~44, no.~6, pp. 2619--2692, 1998.

\bibitem{georgiadis2006resource}
L.~Georgiadis, M.~J. Neely, L.~Tassiulas \emph{et~al.}, ``Resource allocation
  and cross-layer control in wireless networks,'' \emph{Foundations and
  Trends{\textregistered} in Networking}, vol.~1, no.~1, pp. 1--144, 2006.

\bibitem{neely2008fairness}
M.~J. Neely, E.~Modiano, and C.-P. Li, ``{Fairness and optimal stochastic
  control for heterogeneous networks},'' \emph{IEEE/ACM Transactions On
  Networking}, vol.~16, no.~2, pp. 396--409, 2008.

\bibitem{shirani2010mimo}
H.~Shirani-Mehr, G.~Caire, and M.~J. Neely, ``{MIMO downlink scheduling with
  non-perfect channel state knowledge},'' \emph{IEEE Transactions on
  Communications}, vol.~58, no.~7, pp. 2055--2066, 2010.

\bibitem{mo2000fair}
J.~Mo and J.~Walrand, ``Fair end-to-end window-based congestion control,''
  \emph{IEEE/ACM Transactions on networking}, vol.~8, no.~5, pp. 556--567,
  2000.

\bibitem{chen2019dynamic}
Z.~Chen, E.~Bj{\"o}rnson, and E.~G. Larsson, ``{Dynamic resource allocation in
  co-located and cell-free massive MIMO},'' \emph{IEEE Transactions on Green
  Communications and Networking}, vol.~4, no.~1, pp. 209--220, 2019.

\bibitem{8379438}
Z.~Chen and E.~Bj{\"o}rnson, ``{Channel Hardening and Favorable Propagation in
  Cell-Free Massive MIMO With Stochastic Geometry},'' \emph{IEEE Transactions
  on Communications}, vol.~66, no.~11, pp. 5205--5219, 2018.

\bibitem{Huh11}
H.~Huh, A.~M. Tulino, and G.~Caire, ``{Network MIMO with linear zero-forcing
  beamforming: Large system analysis, impact of channel estimation, and
  reduced-complexity scheduling},'' \emph{IEEE Trans. on Inform. Theory},
  vol.~58, no.~5, pp. 2911--2934, 2012.

\bibitem{bethanabhotla2015optimal}
D.~Bethanabhotla, O.~Y. Bursalioglu, H.~C. Papadopoulos, and G.~Caire,
  ``{Optimal user-cell association for massive MIMO wireless networks},''
  \emph{IEEE Trans. on Wireless Commun.}, vol.~15, no.~3, pp. 1835--1850, 2015.

\bibitem{qualcomm2001pfs}
{QUALCOMM, Inc.}, ``{1xEV: 1x EVolution, IS-856 TIA/EIA Standard},'' Tech.
  Rep., Nov. 2001, {Revision 7.2}.

\bibitem{adhikary2013joint}
A.~Adhikary, J.~Nam, J.-Y. Ahn, and G.~Caire, ``{Joint Spatial Division and
  Multiplexing—The Large-Scale Array Regime},'' \emph{{IEEE Trans. on Inform.
  Theory}}, vol.~59, no.~10, pp. 6441--6463, 2013.

\bibitem{9715152}
F.~Ye, J.~Li, P.~Zhu, D.~Wang, H.~Wu, and X.~You, ``{Spectral Efficiency
  Analysis of Cell-Free Distributed Massive MIMO Systems With Imperfect
  Covariance Matrix},'' \emph{IEEE Systems Journal}, vol.~16, no.~4, pp.
  5402--5412, 2022.

\bibitem{9593169}
F.~G{\"o}ttsch, N.~Osawa, T.~Ohseki, K.~Yamazaki, and G.~Caire, ``{The Impact
  of Subspace-Based Pilot Decontamination in User-Centric Scalable Cell-Free
  Wireless Networks},'' in \emph{2021 IEEE 22nd International Workshop on
  Signal Processing Advances in Wireless Communications (SPAWC)}, 2021, pp.
  406--410.

\bibitem{kddi_uldl_precoding}
------, ``{Uplink-Downlink Duality and Precoding Strategies with Partial CSI in
  Cell-Free Wireless Networks},'' in \emph{2022 IEEE Wireless Communications
  and Networking Conference (WCNC)}, 2022, pp. 614--619.

\bibitem{polyanskiy2010channel}
Y.~Polyanskiy, H.~V. Poor, and S.~Verd{\'u}, ``{Channel Coding Rate in the
  Finite Blocklength Regime},'' \emph{IEEE Trans. on Inform. Theory}, vol.~56,
  no.~5, pp. 2307--2359, 2010.

\bibitem{yang2014quasi}
W.~Yang, G.~Durisi, T.~Koch, and Y.~Polyanskiy, ``{Quasi-Static
  Multiple-Antenna Fading Channels at Finite Blocklength},'' \emph{IEEE Trans.
  on Inform. Theory}, vol.~60, no.~7, pp. 4232--4265, 2014.

\bibitem{3gpp38214}
3GPP, ``{NR; Physical layer procedures for data},'' 3GPP Tech. Spec. 38.214, 05
  2022, {Version 17.1.0}.

\bibitem{3gpp38133}
{3GPP}, ``{NR; Requirements for support of radio resource management},'' 3GPP
  Tech. Spec. 38.133, 10 2022, {Version 17.7.0}.

\bibitem{zhang2021local}
J.~Zhang, J.~Zhang, E.~Bj{\"o}rnson, and B.~Ai, ``{Local Partial Zero-Forcing
  Combining for Cell-Free Massive MIMO Systems},'' \emph{IEEE Transactions on
  Communications}, vol.~69, no.~12, pp. 8459--8473, 2021.

\bibitem{huh2011multi}
H.~Huh, S.-H. Moon, Y.-T. Kim, I.~Lee, and G.~Caire, ``{Multi-cell MIMO
  downlink with cell cooperation and fair scheduling: A large-system limit
  analysis},'' \emph{IEEE Transactions on Information Theory}, vol.~57, no.~12,
  pp. 7771--7786, 2011.

\bibitem{3gpp38901}
3GPP, ``{Study on channel model for frequencies from 0.5 to 100 GHz (Release
  16)},'' 3GPP Tech. Spec. 38.901, 12 2019, {Version 16.1.0}.

\end{thebibliography}

\end{document}